\documentclass[copyright,creativecommons]{eptcs}
\providecommand{\event}{ICE'09} 
\providecommand{\volume}{??}
\providecommand{\anno}{2009}
\providecommand{\firstpage}{1}
\providecommand{\eid}{??}

\usepackage{latexsym,amsmath,amssymb,amsfonts,stmaryrd}
\usepackage[mathscr]{eucal}
\usepackage{graphicx}
\usepackage{tabularx}
\usepackage{alltt}
\usepackage{url}
\usepackage{colortbl}
\usepackage{xspace}  
\usepackage{cancel}
\usepackage{semantic}
\usepackage{wasysym}
\usepackage{booktabs} 

\usepackage{tikz}
\usepgflibrary{decorations.pathreplacing} 
\usetikzlibrary{%
  arrows,%
  calc,%
  automata,%
  backgrounds,%
  fit,%
  decorations.pathmorphing%
}

\usepackage{macros}

\def\@definitionautorefname{Definition}

\title{Coordination via Interaction Constraints I:\\
 Local Logic}

\author{Dave Clarke
\institute{Dept.~Computer Science, Katholieke Universiteit Leuven, \\
Celestijnenlaan 200A,\\
    3001 Heverlee, Belgium}
\email{dave.clarke@cs.kuleuven.be}
\and
Jos\'{e} Proen\c{c}a\thanks{Supported by FCT grant 22485 - 2005, Portugal.}
\institute{CWI, \\ Science Park 123, \\ 1098 XG Amsterdam, The Netherlands}
\email{jose.proenca@cwi.nl}
}
\def\titlerunning{Coordination via Interaction Constraints}
\def\authorrunning{D. Clarke \& J. Proen\c{c}a}

\begin{document}

\maketitle

    \begin{abstract}
	Wegner describes coordination as constrained interaction. 
	We take this approach literally and define a coordination model based 
	on \emph{interaction constraints}
	and  partial, iterative and interactive constraint
	satisfaction.
	Our model captures behaviour described in terms of \emph{synchronisation}
	and \emph{data flow constraints}, plus various modes of interaction 
	with the outside world provided
	by	\emph{external constraint symbols}, \emph{on-the-fly constraint}
	generation, and \emph{coordination variables}.
	Underlying our approach is an engine performing 
	(partial) constraint satisfaction of the sets of constraints.		
	Our model extends previous work on three counts: firstly, a more advanced 
	notion of external interaction is offered; secondly, our approach enables
	local satisfaction of constraints with appropriate partial solutions,
	avoiding global synchronisation over the entire constraints set; and,
	as a consequence, constraint satisfaction can finally occur concurrently,
	and multiple parts of a set of constraints can be solved
	and interact with the outside world in an asynchronous manner,  
	unless synchronisation is required by the
	constraints.
	
	This paper describes the underlying logic, which enables a notion of \emph{local solution},
	and relates this logic to the more global approach of our previous work based on 
	classical logic.
    \end{abstract}

\section{Introduction}

Coordination models and languages~\cite{coordination} address the complexity
of systems of concurrent, distributed, mobile and heterogeneous components, by
separating the parts that perform the computation (the components) from the
parts that ``glue" these components together. The glue code offers a layer
between components to intercept, modify, redirect, synchronise communication
among components, and to facilitate monitoring and managing their resource
usage, typically separate from the resources themselves.

Wegner describes coordination as constrained interaction~\cite{wegner}. We
take this approach literally and represent coordination using constraints.
Specifically, we take the view that a component connector specifies a (series
of) constraint satisfaction problems, and that valid interaction between a
connector and its environment corresponds to the solutions of such
constraints.

In previous work~\cite{reo:deconstructing} we took the channel-based coordination model 
\reo~\cite{reo:primer},
extracted  constraints underlying each channel and their composition, and 
formulated  behaviour as a constraint satisfaction problem. There we 
identified that interaction consisted of two phases:  solving  and 
updating constraints.
Behaviour depends upon the current state. The semantics were described
\emph{per-state} in a series of \emph{rounds}. Behaviour in a particular step
is phrased in terms of \emph{synchronisation and data flow constraints}, which
describe the synchronisation and the data flow possibilities of participating
ports. \emph{Data flow} on the end of 
a channel occurs when a single datum is passed through that end.
Within a particular round data flow may occur on some number of ends; this is
equated with the notion of \emph{synchrony}.
The constraints were based on a
synchronisation and  a data flow variable for each port.
Splitting the constraints into 
synchronisation and data flow
constraints is very natural, and it closely resembles the constraint automata
model~\cite{reo:ca06}.
These constraints are solved during the solving phase. Evolution over time is captured
by incorporating state information into the constraints, and updating the
state information between solving phases.
Stronger motivation for the use of constraint-based techniques for the \reo
coordination model can be found in our previous
work~\cite{reo:deconstructing}. 
By abstracting from the channels metaphor and using only the constraints, 
the implementation is free to optimise
constraints, eliminating costly infrastructure, such as unnecessary channels. Furthermore,
constraint-solving techniques are well studied in the literature, and there
are heuristics to search efficiently for  solution, offering 
significant improvement of other  models underlying Reo implementations.
To increase the expressiveness and
usefulness of the model, we added external state variables, external function symbols and
external predicates to the model. These external symbols enable modelling of a wider range of
primitives whose behaviour cannot expressed by constraints, 
either because the internal constraint language is not expressive enough,
or to wrap external entities, such as those with externally maintained state.
The constraint satisfaction process was extended with means for interacting
with external entities to resolve external function symbols and predicates.

In this paper, we make three major contributions to the model: 
\begin{description}
  \item[Partiality] Firstly, we allow solutions for the constraints and the
  set of known predicates and functions to be partial~\cite{partiallogic}. We
  introduce a minimal notion of partial solution which admits solutions only
  on relevant parts (variables) of a connector. External symbols that are only
  discovered on-the-fly are more faithfully modelled in a partial setting.

  \item[Locality] Secondly, we assume a \emph{do nothing} solution for the
  constraints of each primitive exists, where no data is communicated. This
  assumption, in combination with partiality, allows certain solutions for
  part of a connector to be consistently extended to solutions for the full
  connector. Furthermore, our notion of locality 
enables independent parts of the connector to  evolve concurrently.
  
  \item[Interaction] Thirdly, we formalise the constraint satisfaction process
  with respect to the interaction with the external world, and we introduce
  external constraint symbols. These can be seen as lazy constraints, which
  are only processed by the engine on demand, by consulting an external
  source. These can be used to represent, for example, a stream of possible
  choices, which are requested on demand, such as the pages of different
  flight options available on an airline booking web page.
  \end{description}

\paragraph{Organization}
The next section gives an overview of the approach 
taken in this paper, providing a global picture and relating the different 
semantics we present for our constraints. The
rest of the
paper is divided into two main parts. The first part describes how
constraints are defined, and defines four different semantics for variants of the
constraint language and relates them. We present
a classical semantics in \autoref{sec:constrsat} and two partial semantics in \autoref{sec:partiality},
and exploit possible concurrency by allowing local solutions in
\autoref{sec:locality}. The second part introduces a constraint-based engine
to perform the actual coordination, search and applying solutions for the
constraints. We describe stateful primitives in \autoref{sec:state}, and 
add interaction in
\autoref{sec:interaction}. We give some conclusions about this work in
\autoref{sec:conclusion}.

\section{Coordination = Scheduling + Data Flow}

We view coordination as a constraint satisfaction problem, where solutions for the
constraints yield how data should be communicated among components. More
specifically, solutions to the constraints describe \emph{where} and
\emph{which} data flow. Synchronisation variables describe the where, and
data flow variables describe the which.
With respect to our previous work~\cite{reo:deconstructing}, we move from a
classical semantics to a local semantics, where solutions address only part
of the connector, as only a relevant subset of the variables of the
constraints are required for solutions.
We do this transformation for classical to local semantics
in a stepwise manner, distinguishing four different
semantics that yield different notions of valid solution $\sigma$, 
mapping synchronisation and data flow variables to appropriate values:
\begin{description}
  \item[Classical semantics] ~     \begin{itemize}
      \item $\sigma$ are always total (for the variables of the connector under consideration);
      \item an explicit value \NOFLOW\ is added to the data domain to
      represent the data value when there is no data flow;
      \item an explicit \emph{flow axiom} is added to constraints to ensure the proper
      relationship between synchronisation variables and data flow variables; and
      \item constraints are solved globally across the entire `connector'.
    \end{itemize}
    
  \item[Partial semantics] ~     \begin{itemize}
      \item $\sigma$ may be partial, not binding all variables in a constraint;
      \item the \NOFLOW\ value is removed and modelled by leaving the data flow
      variable undefined; and       
	\item as the previous flow axiom is no longer expressible,  the relationship
      between synchronisation and data flow variables is established by a new
      meta-flow axiom, which acts after constraints have been solved to  filter invalid solutions.
    \end{itemize}
    
  \item[Simple semantics] ~
    \begin{itemize}
      \item $\sigma$ is partial, and the semantics is such that only certain
      ``minimal" solutions,  which define only the necessary
      variables, are found; and
      \item the meta-flow axiom is expressible in this logic, so a \emph{simple} flow axiom can again be added to the constraints.

    \end{itemize}
    
  \item[Local semantics] ~     \begin{itemize}
      \item formul\ae\ are partitioned       into
        blocks, connected via shared variables;
      \item each block is required to always admit a \emph{do nothing} 
      solution;
      \item some solutions in a block can be found without looking at its
      neighbours, whenever there is no-flow on its \emph{boundary}
      synchronisation variables;
      \item two or more such solutions  are locally  compatible;
      \item blocks can be merged in order to find more solutions,
        in collaboration, when existing solutions do not ensure the no-flow condition over the \emph{boundary} synchronisation variables; and
      \item the search space underlying constraints is smaller
      than in the previous semantics, and there is a high degree of
      locality and concurrency.

    \end{itemize}
\end{description}

We present formal embeddings between these logics, with respect to solutions 
that obey the various (meta-) flow axioms (linking 
solutions for synchronisation and data flow variables). 
We call such solutions \emph{firings}.
The first step is from a classical to a partial semantics. The number of
solutions increases, as new (smaller) solutions also become  valid. We then
move to a simple semantics to regain an expressible flow axiom, where only
some ``minimal" partial solutions are accepted. In the last step we present a
local semantics, where we avoid the need to inspect even more constraints,
namely, we avoid visiting constraints added conjunctively to the system, by
introducing some requirements on solutions to blocks of constraints.

\section{Coordination via Constraint Satisfaction}
\label{sec:constrsat}

In previous work we described coordination in terms of constraint satisfaction.
The main problem with that approach is that the constraints needed to be solved
globally, 
which means that it is not scalable as the basis of an engine for 
coordination. In this section, we adapt the underlying logic and notion
of solution to increase the amount of available locality and concurrency in the constraints.
Firstly, we move from the standard classical interpretation of the logic
to a partial interpretation. This offers some improvement, but the solutions
of a formula need to be filtered using a semantic variant of the flow axiom,
which is undesirable  because filtering them out during the constraint satisfaction 
process could be significantly faster.
We improve on this situation by introducing a simpler notion of solution for formul\ae,
requiring only relevant variables to be assigned. 
This approach avoids post hoc filtering of solutions.
 Unfortunately, even simple solutions still require more
or less global constraint satisfaction. 
Although it is the case that  many constraints
may correspond to no behaviour within parts of the connector---indeed 
all constraints admit such solutions---, 
the constraint satisfier must still visit the constraints to determine
this fact. In the final variant, we simply assume that the no behaviour solution
can be chosen for any constraint not visited by the constraint solver,
and thus the constraint solver can find solutions to constraints without
actually visiting all constraints. This means that more concurrency is
available and different parts of the implicit constraint graph can be
solved independently and concurrently.

We start by motivating our work via an example, and we then describe the
classical approach to constraint satisfaction and its problems, before
gradually refining our underlying logic to a be more amenable to
scalable satisfaction.

\subsection{Coordination of a complex data generator}
\label{sec:example}

We introduce a motivating example, depicted in Figure~\ref{fig:cdg}, where a
\emph{Complex Data Generator (CDG)} sends data to \emph{Client}. Data
communication is controlled via a coordinating connector. The connector
consists of a set of composed coordination building blocks, each with some
associated constraints describing their behavioural possibilities. 
We call these building blocks simply
\emph{primitives}.
The CDG and the Client are also primitives, and play the same coordination
game by providing some constraints reflecting their intentions to read or 
write data.

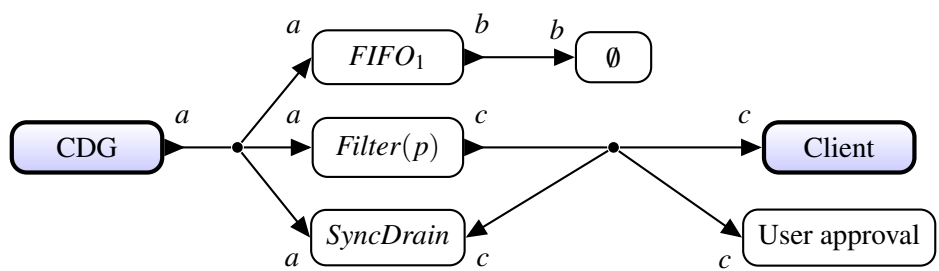
\begin{figure}[htb]
\centering
\begin{tikzpicture}
  [node distance=3cm,
      >=triangle 45,thick,
   dot/.style={circle,inner sep=0.5mm,fill=black},
   height/.style={node distance=1.2cm},
   closer/.style={node distance=2cm}
  ]    
  \node[main,label=10:$a$]         (cdg)    {CDG};
  \node[dot,closer,right of=cdg]  (cdg2)   {};
  \node[prim,closer,right of=cdg2,label=170:$a$,label=10:$c$]
    (filter)   {\filterp};
  \node[prim,height,above of=filter,label=170:$a$,label=10:$b$]
    (fifo)   {\fifo};
  \node[prim, right of=fifo,label=170:$b$,minimum width=10mm]
    (null)   {$\emptyset$};
  \node[prim,height,below of=filter,label=184:$a$,label=-4:$c$]
    (sd)     {\syncdrain};
  \node[dot,right of=filter] (filter2) {};
  \node[main,right of=filter2,label=170:$c$] (res) {Client};
  \node[prim,height,below of=res,label=184:$c$] (user) {User approval};
  
  \draw[>-] (cdg) to (cdg2);
  \draw[->] (cdg2) to (fifo.west);
  \draw[->] (cdg2) to (filter.west);
  \draw[->] (cdg2) to (sd.west);
  \draw[>->] (fifo) to (null);
  \draw[>-] (filter) to (filter2);
  \draw[->] (filter2) to (sd.east);
  \draw[->] (filter2) to (user.west);
  \draw[->] (filter2) to (res.west);

\end{tikzpicture}

\caption{Network of constraints: coordinating a complex data generator.}
\label{fig:cdg}
\end{figure}

\autoref{fig:cdg} uses a graphical notation to depict how the different
primitives are connected. Each box represents a primitive with some associated
constraints, connected to each other via shared variables. For example, the
CDG shares variable $a$ with \fifo, \filterp, and \syncdrain, indicating
that the same data flows through the lines connecting theses primitives
in the figure.
The arrows represent the direction of data flow, thus the value of $a$ is given
by $CDG$ and further constrained by the other attached primitives. 

Most of the coordination primitives are channels from the
\reo coordination
language~\cite{reo:primer}. 
Previous work~\cite{reo:deconstructing}  
described a constraint-based approach to modelling \reo networks. 
Here we forego the graphical representation to emphasise the idea that
a coordinating connector can be seen as a soup of constraints linked by shared variables.
One optimisation we immediately employ is using the same name for ends which are
connected by synchronous channels or replicators.\footnote{Semantically, this view of synchronous channels
and replicators is valid.} Note that in the original description of \reo, nodes
act both as mergers and replicators. This behaviour can be emulated using
merger and replicator primitives, as we have done. The result is a simpler
notion of node, a 1:1 node which both synchronises and copies data from the
output port to the input port.
Primitives act as constraint providers, which are combined until they reach a
consensus regarding where and which data will be communicated. Only then a
possible communication of data takes place, and the primitives update their
constraints.

In the particular case of the example in \autoref{fig:cdg}, there is a
complex data generator (CDG) that can write one of several possible values,
a filter that can only receive (and send) data if it validates a given
predicate $p$, a component (User approval) that approves values based on 
interaction with a user, a destination client that receives the final result, 
and some other primitives that impose further constraints. We will come back to this example
after introducing some basic notions about the constraints.

\paragraph{Notation}
We write $\Data$ to denote a global set of possible data that can flow on the
network. \NOFLOW is a special constant not in $\Data$ that represents no
data flow. \X denotes a set of \emph{synchronisation variables} over $\set{\true,\false}$, $\wh\X =
\set{\wh{x} ~|~ x \in \X}$ a set of \emph{data flow variables} over $\Data \cup \set{\NOFLOW}$, $\dvar{P}$ a
set of predicate symbols, and $\dvar{F}$ a set of function symbols such that 
$\Data\subseteq \dvar{F}$. (Actually, $\Data$ is the Herbrand universe over function
symbols $\dvar{F}$.)
We use the
following variables to range over various domains:
$x\in\X$,
$\wh{x}\in\wh\X$, 
$f\in\dvar{F}$, and 
$p\in\dvar{P}$.
Recall that synchronisation variables $\X$ and data flow variables $\wh\X$ are intimately
related, as one describes whether data flows and the other describes what the data is.

\subsection{Classical Semantics}
\label{sec:classical}

Consider the logic with the following syntax of formul\ae\ ($\psi$) and
terms ($t$):
\begin{eqnarray*}
	\psi & ::= &
 \true ~~|~~ x ~~|~~ \psi_1 \land \psi_2 ~~|~~ \lnot \psi ~~|~~ p(t_1,\ldots,t_n)  \\
  t & ::= & \wh{x} ~~|~~ f(t_1,\ldots,t_n)
\end{eqnarray*}

\noindent
$\True$ is \emph{true}. We assume that one of the internal predicates in
$\dvar{P}$ is equality, which is denoted using the standard infix notation
$t_1 = t_2$. The other logical connectives can be encoded as usual:
$\False\df\neg\True$; 
$\psi_1\lor\psi_2\df\neg(\neg\psi_1\land\neg\psi_2)$;
$\psi_1\rightarrow\psi_2\df\neg\psi_1\vee\psi_2$; and
$\psi_1\leftrightarrow\psi_2\df(\psi_1\rightarrow\psi_2)\wedge 
(\psi_2\rightarrow\psi_1)$. 
Constraints can be easily extended with an existential quantifier, provided
that it does not appear in a negative position, or alternatively, that
it is used only at the top level.

The semantics is based on a relation $\sigma,\I\cmodels\psi$, where
$\sigma$ is a total map from $\dvar{X}$ to $\set{\true,\false}$ and
from $\wh{\dvar{X}}$ to $\Data\cup\set{\NOFLOW}$,
and $\I$ is an arity-indexed total map from $\dvar{P}_n\times\T^{n}$ to 
$\set{\true,\false}$, for each $n\ge 0$, where
$\dvar{P}_n$ is the set of all predicate symbols of arity $n$,
 \T is the set of all possible ground terms
(terms with no variables) plus the constant \NOFLOW. 
The semantics is defined by a satisfaction relation $\cmodels$ defined as 
follows. The function $\mathit{Val}_{\sigma}$ replaces all variables $v$ by 
$\sigma(v)$, and we assume that $\mathit{Val}_{\sigma}(f(t_1,\ldots,t_n)) = \NOFLOW$ whenever $t_i=\NOFLOW$, for some $i\in 1..n$.

\begin{definition}[Classical Satisfaction]
\[
\begin{array}{lcl}
  \sigma,\I \cmodels \True    & \multicolumn{2}{l}{always}\\
  \sigma,\I\cmodels x & \ifff & \sigma(x)=\true \\
  \sigma,\I\cmodels \psi_1\land\psi_2
    &\ifff& \sigma,\I\cmodels\psi_1 \mbox{ and }
            \sigma,\I\cmodels\psi_2 \\
  \sigma,\I\cmodels\lnot\psi & \ifff & \sigma,\I\notcmodels\psi \\
  \sigma,\I\cmodels p(t_1,\ldots,t_n) & \ifff &
        p(\mathit{Val}_{\sigma}(t_1),
        \ldots,\mathit{Val}_{\sigma}(t_n))\mapsto\true\in\I
\end{array}
\]
\end{definition}

The following axiom relates synchronisation and data flow variables,
stating that a synchronisation variable being set to  $\false$ 
corresponds to the corresponding data flow being $\NOFLOW$.

\begin{axiom}[Flow Axiom]
	\begin{equation}
	\lnot x \leftrightarrow \widehat{x} = \NOFLOW
	\tag{flow axiom}
	\end{equation}
\end{axiom}

We introduced  this in our previous approach for coordination via
constraints~\cite{reo:deconstructing}. Every pair
of variables, $x$ and $\wh{x}$, is expected to obey this axiom.
Write $FA(x)$ for the flow axiom for variables $x,\wh{x}$ and
$FA(X)$ for the conjunction $\bigwedge_{x\in X}FA(x)$. Also write $\fv{\psi}$ 
for the free variables of $\psi$, i.e., variables from \X and $\wh\X$ that 
occur in $\psi$.

\begin{definition}[Classical Firing]
	A solution $\sigma$ to constraint $\psi$
	which satisfies the meta-flow axiom
	is called a \emph{classical firing}.
	That is,
	$\sigma$ is a classical firing for $\psi$ if and only if 
	 $\sigma,\I\cmodels \psi\wedge FA(\fv{\psi})$.~
\end{definition}

\newcommand{\myprim}[1]{
\wrap{\begin{tikzpicture}  [node distance=5mm,
   >=triangle 45,thick,
   text height=2mm
  ]
  #1
\end{tikzpicture}
}}

\begin{example}
Recall the example from \autoref{sec:example}. We define the constraints for
each primitive in \autoref{tab:statelessex}. The client does not impose any
constraints on the input data ($\true$), and the \fifo primitive is empty so its
constraints only say that no data can be output. $\mathit{UserAppr}$ is an
external predicate symbol, which must be resolved using external interaction (See
\S~\ref{sec:externallogic}).
Later we extend some of these constraints to capture the notion of state and interaction
(see Table~\ref{tab:interactiveex}). 
The behaviour of the full system is given by the firings for the
conjunction of all constraints.

\begin{table}[htb]
\centering
$\begin{array}{clc}
\toprule
~\hspace{15mm} Primitive \hspace{15mm}~ & \multicolumn{2}{c}{Constraint} \\
\cmidrule(lr){1-1}\cmidrule(lr){2-3}
\myprim{
  \node at (-1.5,0) {~};
  \node[main] (x) at (0,0) {CDG};
  \node at (1.5,0) {a} edge[<-] (x);}
& \psi_1 = &     {a \to (\widehat{a} = \mathtt{d_1} \lor \widehat{a} = \mathtt{d_2} \lor
     \widehat{a} = \mathtt{d_3})} \\
\myprim{
  \node at (-1.5,0) {c} edge [->] (x);
  \node[main] (x) at (0,0) {Client};
  \node at (1.5,0) {~};}
& \psi_2 = &
    { \true }
\\
\myprim{
  \node[prim] (x) at (0,0) {User approval};
  \node at (-1.8,0) {c} edge[->] (x);
  \node at (1.8,0) {~};}
& \psi_3 = &     {c \to \mathit{UserAppr}(\widehat{c})}
\\
\myprim{
  \node[prim] (x) at (0,0) {\syncdrain};
  \node at (-1.5,0) {a} edge[->] (x);
  \node at (1.5,0) {c} edge[->] (x);}
& \psi_4 = &     {a \leftrightarrow c}
\\
\myprim{
  \node[prim] (x) at (0,0) {\filterp};
  \node at (-1.5,0) {a} edge[->] (x);
  \node at (1.5,0) {c} edge[<-] (x);}
& \psi_5 = &
    {\begin{array}{cc}
      c \to a \land       c \to (p(\widehat{c}) \land \widehat{a} = \widehat{c})
      \\ ~\land (a\land p(\wh{a})) \to c
     \end{array}
    }
\\
\myprim{
  \node at (-1.5,0) {a} edge[->] (x);
  \node[prim] (x) at (0,0) {\fifo};
  \node at (1.5,0) {b} edge[<-] (x);}
& \psi_6 = &     {\lnot b } \\
\myprim{
  \node[prim,minimum width=10mm] (x) at (0,0) {$\emptyset$};
  \node at (-1,0) {b} edge[->] (x);
  \node at (1,0) {~};}
& \psi_7 = &     {\lnot b}
\\
\bottomrule
\end{array}$
\caption{List of primitives and their associated constraints, where $\mathtt{d_1},\mathtt{d_2},\mathtt{d_3} \in \Data$.}
\label{tab:statelessex}
\end{table}
\end{example}

Consider the \filterp primitive in \autoref{tab:statelessex}. The flow axiom is given by $FA(a,c)$. The constraint
$c \to a$ can be read as: if there is data flowing on $c$, then there must
also be data flowing on $a$. The second part, $c \to (p(\widehat{c}) \land
\widehat{a} = \widehat{c})$, says that when there is data flowing on $c$ its
value must validate the predicate $p$, and data flowing in $a$ and $c$ must be
the same. Finally, the third part $(a\land p(\wh{a})) \to c$ states that data
that validates the predicate $p$ cannot be lost, i.e., flow on $a$ but not on
$c$. A classical firing for the interpretation \I is $\set{a\mapsto\true,
c\mapsto\true, \wh{a}\mapsto d, \wh{c}\mapsto d}$ whenever $d \in \Data$ is such that
$p(d)\mapsto\true \in \I$. The assignment $\set{c\mapsto\true, \wh{c}\mapsto
\NOFLOW}$ is not a classical firing because it violates the flow axiom, and
because it is not a total map (it refers to neither  $a$ nor $\wh{a}$).

\section{Partiality}
\label{sec:partiality}

The first step towards increasing the amount of available concurrency and
the scalability of our approach is to make the logic partial. This means that
solutions no longer need to be total, so for some $x\in\dvar{X}$ or
$\wh{x}\in\wh{\dvar{X}}$,
$\sigma(x)$ or $\sigma(\wh{x})$ may not be defined. 
In addition, we drop the $\NOFLOW$ value and so $\sigma(\wh{x})$ may 
either map to a value from $\Data$ or be undefined.
The semantics is defined by a satisfaction relation 
$\pmodels$ and a dissatisfaction relation $\copmodels$ defined below.
These state, respectively, when a formula is definite true 
or definitely false.
Partiality is introduced in either the clause for $x$
or for $p(t_1,\ldots,t_n)$, whenever some variable is not
defined in $\sigma$. An assignment $\sigma$ now is a \emph{partial} map 
from synchronisation variables to
$\set{\true,\false}$, and from data flow variables to $\Data$.
Similarly, an interpretation $\I$ is now an arity-indexed 
family of partial map from
 $\dvar{P}_n\times\T^{n}$ to
$\set{\true,\false}$, where \T is the set of all possible ground
terms, and $\mathit{Val}_{\sigma}(f(t_1,\ldots,t_n))=\bot$ whenever
$\mathit{Val}_{\sigma}(t_i)=\bot$, for some $i\in 1..n$. We use $\bot$ to
indicate when such a map is undefined.

\begin{definition}[Partial Satisfaction]
\[
\begin{array}{lcl}
  \sigma,\I \pmodels \True    & \multicolumn{2}{l}{always}\\
  \sigma,\I\pmodels x & \ifff & \sigma(x)=\true \\
  \sigma,\I\pmodels \psi_1\land\psi_2
    &\ifff& \sigma,\I\pmodels\psi_1 \mbox{ and }
            \sigma,\I\pmodels\psi_2 \\
  \sigma,\I\pmodels\lnot\psi & \ifff & \sigma,\I\copmodels\psi \\
  \sigma,\I\pmodels p(t_1,\ldots,t_n) & \ifff &
        p(\mathit{Val}_{\sigma}(t_1), \ldots,\mathit{Val}_{\sigma}(t_n))\mapsto\true\in\I
       \\
\\
  \sigma,\I\copmodels x & \ifff & \sigma(x)=\false \\
  \sigma,\I\copmodels \psi_1\land \psi_2 & \ifff & 
     \sigma,\I\copmodels \psi_1 \mbox{ or } \sigma,\I\copmodels \psi_2 \\
  \sigma,\I\copmodels \lnot \psi & \ifff & \sigma,\I\pmodels \psi \\
  \sigma,\I\copmodels p(t_1,\ldots,t_n) & \ifff &
     p(\mathit{Val}_{\sigma}(t_1), \ldots,\mathit{Val}_{\sigma}(t_n))\mapsto\false\in\I
\end{array}
\]
\end{definition}

\begin{lemma} \label{lem:totalsol} 
	If $\sigma$ and $\I$ are total, then
	either $\sigma,\I\pmodels \psi$ or $\sigma,\I\copmodels \psi$, but it is 
	never undefined.
\end{lemma}

We need to adapt the flow axiom as it refers explicitly to $\NOFLOW$, which
is no longer available. The obvious change would be to replace
$\NOFLOW$ by partiality, giving (semantically) 
$\sigma(x)=\false \Longleftrightarrow \sigma(\widehat{x}) = \bot$.
But we can do better, permitting $\sigma(x)=\bot$ to also represent 
no data flow. In addition, it is feasible that $\sigma(x)=\true$
with $\sigma(\wh{x})=\bot$ is a valid combination, to cover the case
where the actual value of the data does not matter.
Together, these
give the following \emph{meta-flow axiom}, which is a semantic and not 
syntactic condition.

\begin{axiom}[Meta-Flow Axiom]
	An assignment $\sigma$ obeys the \emph{meta-flow axiom} 
	whenever for all $x\in\dvar{X}$:
	\begin{eqnarray*}
	\sigma(\widehat{x})\neq\bot & \Longrightarrow & \sigma(x) = \True
	\end{eqnarray*}
	Write $MFA(\sigma)$ whenever $\sigma$ obeys the meta-flow axiom.
\end{axiom}

The following table gives all solutions to the meta-flow axiom:

\begin{center}
	\begin{tabular}{c|cccc|cc}
		& \multicolumn{4}{c|}{possible} & \multicolumn{2}{c}{forbidden} \\ \hline
	$x$ &  $\true$ & $\true$ & $\false$ & $\bot$ & $\false$ & $\bot$ \\
	$\wh{x}$ & $d$ & $\bot$ & $\bot$ & $\bot$ & $d$ & $d$ 
	\end{tabular}
\end{center}

For comparison, the following table gives the solutions for the flow axiom:
\begin{center}
	\begin{tabular}{c|cc|cc}
		& \multicolumn{2}{c|}{possible} & \multicolumn{2}{c}{forbidden} \\ \hline
	$x$ &  $\true$ & $\false$ &  $\true$ & $\false$ \\
	$\wh{x}$ & $d$ & $\NOFLOW$ & $\NOFLOW$ & $d$ 
	\end{tabular}
\end{center}

\begin{definition}[Partial Firing]
	A partial solution $\sigma$ to a constraint $\psi$ that satisfies the meta-flow axiom is called a 
	\emph{partial firing}. That is,
	whenever $\sigma,\I\pmodels \psi$ and $MFA(\sigma)$.
\end{definition}

Consider again the constraints of the $\filterp$ primitive in
\autoref{tab:statelessex}. A possible firing for it is
$\set{a\mapsto\true, c\mapsto\false, \wh{a}\mapsto d, \wh{c}\mapsto \NOFLOW}$
where $d \in \Data$ does not validate the predicate $p$. The
equivalent partial solution can be obtained by replacing \NOFLOW by $\bot$.
Therefore, $\set{a\mapsto\true,c\mapsto\false,\wh{a}\mapsto d}$ is
also a partial firing, whenever $p(d)$ does not hold. Note also that
$\set{c\mapsto\false},\I\pmodels a \to c$ holds in the partial setting, yet
$\set{c\mapsto\false},\I\cmodels a \to c$ does not hold in the classical
setting, because the classical satisfaction requires the solutions to be total
mappings of all variables involved.

\subsection[Embeddings: Classical and Partial]{Embeddings: Classical $\leftrightarrow$ Partial}

We can move from an explicit representation of no-flow, namely
$\sigma(x)=\false$ and $\sigma(\wh{x})=\NOFLOW$, to 
an implicit representation using partiality, namely either
$\sigma(\wh{x})=\bot$ and $\sigma(x)=\bot$ or 
$\sigma(x)=\false$, which means
that the constraint solver need not find a value for  $x$ or $\wh{x}$.

\begin{lemma}[Classical to Partial]
  Let $\psi$ be a constraint where $\NOFLOW$ does not occur in $\psi$,
  and $\sigma$ be an assignment where $\dom(\sigma)=\fv{\psi}$. We write
  $\I^\circ$ to represent the interpretation obtained by replacing in \I the
  constant \NOFLOW by $\bot$.
	If $\sigma$ is a classical firing for $\psi$ and the interpretation \I, then
	$\sigma^\circ$ is a partial firing for $\psi$ and the interpretation 
	$\I^\circ$, where $\sigma^\circ$ is defined as follows:
\[\begin{array}{lcl}
	\sigma^\circ(x) & = & \sigma(x) \\
	\sigma^\circ(\wh{x}) &=& \left\{  
		\begin{array}{ll}
			\bot, & \mbox{if }\sigma(\wh{x})=\NOFLOW  \\
			\sigma(\wh{x}), & \mbox{otherwise}
		\end{array}
		\right.
\end{array}
\]
\end{lemma}

\begin{proof}
  Assume that $\sigma,\I\cmodels\psi \land FA(\fv{\psi})$. Then  (1)
  $\sigma,\I\cmodels\psi$ and  (2) $\sigma,\I\cmodels FA(\fv{\psi})$. It can
  be seen by straightforward induction that (1) implies that
  $\sigma^\circ,\I^\circ\pmodels\psi$, because $\I^\circ$ maps the same values than $\I$ after replacing \NOFLOW by $\bot$, and $\sigma^\circ$ is defined for all free synchronisation variables. Since (2) holds for every $x\in
  \fv{\psi}$ and $\dom(\sigma^\circ) \cap \X = \dom(\sigma)  \cap \X = \fv{\psi}  \cap \X$, we can
  safely conclude that $MFA(\sigma^\circ)$: when $\sigma^\circ(x)\neq\true$ 
  then
  $\sigma^\circ(x)=\sigma(x)=\false$, implying by the flow axiom that
  $\sigma(\wh{x})=\NOFLOW$, where we conclude that $\sigma^\circ(\wh{x}) =
  \bot$ (and the meta-flow axiom holds).
\end{proof}

\begin{lemma}[Partial to Classical]
  If $\sigma$ is a partial firing for $\psi$ and for a total interpretation \I, then
  $\sigma^\dag$ is a classical firing for $\psi$ and for an interpretation
  $\I^\dag$, where $\I^\dag$ results from replacing $\bot$ by \NOFLOW in \I,
  and $\sigma^\dag$ is defined as follows:
  \[
		\begin{array}{rcl@{\qquad}l}
		\sigma^\dag(x) & = & \sigma(x), &  \mbox{if }\sigma(x)\neq\bot \\
		\sigma^\dag(x) & = & \false,  & \mbox{if }\sigma(x)=\bot \\ \\
		\sigma^\dag(\wh{x}) & = & \sigma(\wh{x}), & \mbox{if } \sigma(\wh{x})\neq\bot \\
		\sigma^\dag(\wh{x}) & = & \NOFLOW, & \mbox{if } \sigma(\wh{x})=\bot 
						\mbox{ and }\sigma(x)\neq\true 	 \\
		\sigma^\dag(\wh{x}) & = & 42, & \mbox{if } \sigma(\wh{x})=\bot 
						\mbox{ and }\sigma(x)=\true
		\end{array}
	\]
\end{lemma}

\begin{proof}
  Assume that (1) $\sigma,\I \pmodels \psi$ and (2) $MFA(\sigma)$. Note that
  $\sigma^\dag$ is total, $\I^\dag$ is total, and $\sigma \subseteq
  \sigma^\dag$. It can be seen by \autoref{lem:totalsol} that
  $\sigma\subseteq\sigma^\dag$ and (1) imply
  $\sigma^\dag,\I^\dag\cmodels\psi$. We show that $\sigma^\dag,\I^\dag\cmodels FA(\fv{\psi})$
  by considering all possible cases. (i) If $\sigma(\wh{x}) \neq \bot$ then
  $\sigma^\dag(\wh{x}) = \sigma(\wh{x})$, and from (2) we conclude that
  $\sigma^\dag(x) = \sigma(x) = \true$. (ii) If $\sigma(\wh{x}) = \bot$ and
  $\sigma(x)=\true$, then $\sigma^\dag(\wh{x})=42$ and $\sigma^\dag(x)=\true$.
  (iii) If $\sigma(\wh{x}) = \bot$ and $\sigma(x)\neq\true$, then
  $\sigma^\dag(x)=\false$ and $\sigma^\dag(\wh{x})=\NOFLOW$.
\end{proof}

\subsection{Inexpressibility of Meta-Flow Axiom}

Unfortunately, the meta-flow axiom is not expressible in partial logic.
A consequence of this is that if partial logic is used 
as the constraint language, solutions may be found which do not
satisfy this axiom; such solutions subsequently 
need to be filtered after performing
constraint satisfaction, which clearly is not ideal, as 
the constraint engine would need to continue to find a real solution,
having wasted time finding this non-solution.

The following lemma will help prove that the meta-flow axiom is not expressible.

\begin{lemma}\label{lemma:subsetpartial}
	If $\sigma,\I\pmodels\psi$ and $\sigma\subseteq\sigma'$, then $\sigma',\I\pmodels\psi$.
\end{lemma}
\begin{proof} By straightforward induction on $\psi$.
\end{proof}

\begin{lemma}\label{lemma:nonpreserve} 
	No formula $\psi$ exists such that 
	$\sigma,\I\pmodels\psi$ if and only if $MFA(\sigma)$.
\end{lemma}
\begin{proof}
Assume that $\psi_{MFA}$ is such a formula over variables
$\set{x,\wh{x}}$. Then for $\sigma=\set{x\mapsto\false}$ 
we have that $\sigma,\I\pmodels\psi_{MFA}$. 
Now $\sigma\subseteq\sigma'=\set{x\mapsto\false,\widehat{x}\mapsto 42}$. 
Hence by \autoref{lemma:subsetpartial}, we have that $\sigma',\I\pmodels\psi_{MFA}$.
But $\sigma'$ does not satisfy the meta-flow axiom. Contradiction.
\end{proof}

\subsection{Simple Logic}

Using partial logic as the basis of a coordination engine  is not ideal,
as constraint satisfaction for this logic could find solutions which
do not satisfy the meta-flow axiom (due to \autoref{lemma:nonpreserve}).
Such solutions would need to be filtered in a post-processing phase,
resulting in an undesirable overhead.

We resolve this problem by 
modifying the semantics so that only certain
`minimal' solutions are found. These solutions define only the necessary
variables---which has the consequence that the constraint solver needs only
to satisfy variables mentioned in the (relevant branch of a) constraint. We
extend also the syntax of formul\ae\ by distinguishing two kinds of
conjunctions. The \emph{overlapping conjunction} ($\land$) of two constraints 
accepts two compatible solutions and joins them together, while an \emph{additive conjunction} 
($\sland$) accepts only solutions which satisfy both constraints. 
Both kinds of conjunction are present, firstly, to talk about 
the joining of solutions for (partially) independent parts
of a connector (overlapping conjunction), and to enforce overarching constraints,
such as the flow axiom (additive conjunction).
The semantics for the logic is formalised in
Definition~\autoref{def:simplesat}.
In this logic, the meta-flow axiom is expressible.

\begin{definition}[Simple satisfaction]
\label{def:simplesat}
We define inductively a simple satisfaction relation $\sigma,\I\smodels\psi$ and a simple disatisfaction relation $\sigma,\I\cosmodels\psi$, where the assignment $\sigma$ and the interpretation \I may be partial.
\[
\begin{array}{@{}r@{ }c@{ }lcl@{}}
  \emptyset,\I &\smodels& \True    & \multicolumn{2}{l}{always}\\
  \set{[x \mapsto \true]},\I&\smodels& x & \multicolumn{2}{l}{always} \\
  \sigma_1 \cup \sigma_2,\I&\smodels& \psi_1\land\psi_2
    &\ifff& \sigma_1,\I\smodels\psi_1 \mbox{ and }
            \sigma_2,\I\smodels\psi_2 \mbox{ and }
            \sigma_1 \frown \sigma_2\\
  \sigma,\I&\smodels& \psi_1\sland\psi_2
    &\ifff& \sigma,\I\smodels\psi_1 \mbox{ and }
            \sigma,\I\smodels\psi_2\\
  \sigma,\I&\smodels&\lnot\psi & \ifff & \sigma,\I\cosmodels\psi \\
  \sigma,\I&\smodels& p(t_1,\ldots,t_n) & \ifff &
        p(\mathit{Val}_{\sigma}(t_1), \ldots,\mathit{Val}_{\sigma}(t_n))\mapsto\true\in\I
        \mbox{ and }
        dom(\sigma) = \mathit{fv}(p(t_1\ldots,t_n)) \\
\\
  \set{[x\mapsto \false]},\I&\cosmodels& x & \multicolumn{2}{l}{always}\\ 
  \sigma,\I&\cosmodels& \psi_1\land \psi_2 & \ifff & 
	  \mbox{for all }\sigma_1,\sigma_2 \mbox{ s.t. } 
	\sigma_1\frown\sigma_2 \mbox{ and }\sigma=\sigma_1\cup\sigma_2
		\\&&&&
		\mbox{we have }
      \sigma_1,\I\cosmodels \psi_1 \mbox{ or } \sigma_2,\I\cosmodels \psi_2 \\
  \sigma,\I&\cosmodels& \psi_1\sland \psi_2 & \ifff & 
     \sigma,\I\cosmodels \psi_1 \mbox{ or } \sigma,\I\cosmodels \psi_2 \\
  \sigma,\I&\cosmodels& \lnot \psi & \ifff & \sigma,\I\smodels \psi \\
  \sigma,\I&\cosmodels& p(t_1,\ldots,t_n) & \ifff &
     p(\mathit{Val}_{\sigma}(t_1), \ldots, \mathit{Val}_{\sigma}(t_n))\mapsto\false\in\I
     \mbox{ and }
     dom(\sigma) = \mathit{fv}(p(t_1\ldots,t_n)) 
\end{array}
\]
\[
\mbox{where~~} \begin{array}{rcl}
  \sigma_1 \frown \sigma_2 &\defn&
  \forall x \in dom(\sigma_1) \cap dom(\sigma_2). \sigma_1(x) = \sigma_2(x) \\
\end{array}\]
\end{definition}

The additive conjunction $\sland$ of $\psi_1$ and $\psi_2$ is satisfied by
$\sigma$ if $\sigma$ satisfies both $\psi_1$ and $\psi_2$. The overlapping
conjunction $\land$ is more relaxed, and simply merges any pair of solutions
for $\psi_1$ and $\psi_2$ that do not contradict each other. For the
constraints of the primitives, the conjunctions that appear in a positive
position are regarded as overlapping conjunctions ($\land$), while the
conjunctions that appear in a negative position are regarded as
additive conjunctions ($\sland$).%
\footnote{A positive position is inside the scope of an even number of negations, and a negative position is inside the scope of an odd number of negations. For example, in $(\lnot (a \land \lnot b)) \land c$, $a$ is in a negative position, while $b$ and $c$ are in a positive position.}
As a consequence, the rule for $\sigma,\I\smodels
\psi_1\sland\psi_2$ is only used when applying the flow axiom, as we will 
soon see, and the rule
for $\sigma,\I\cosmodels \psi_1\land\psi_2$ is present mainly for the
completeness of the definition.

\paragraph{Notation}
In the following we write $\psi^S$ to 
represent the constraints 
obtained by replacing all conjunctions in $\psi$ in negative positions by
$\sland$, and we write $\psi^P$ to represent the constraints obtained by replacing 
all occurrences of $\sland$ in $\psi$  by $\land$. We also encode $\psi_1 \slor \psi_2$ as $\lnot ( \lnot \psi_1 \sland \lnot \psi_2)$.

\paragraph{~}
When specifying constraints in simple logic, 
we never use $-\lor-$ in a positive position, which would correspond to
$\neg(\neg-\land\neg-)$, as this means satisfying the clause 
$\sigma,\I \cosmodels \psi_1 \land \psi_2$ in order
to find a given assignment.  Therefore we do not 
require the use of universal quantification over solution sets. 
In the partial satisfaction relation, we define how to verify that a given
pair $\sigma,\I$ satisfies a constraint. The simple satisfaction relation aims
at \emph{constructing} $\sigma$ such that the pair $\sigma,\I$ satisfies the
constraints. Assuming the universal quantifier is never used, we believe that
the simple satisfaction relation describes a constructive process to obtain a
solution that is not more complex than searching for a solution in partial 
logic.
Note that we still lack experimental verification of this intuition.

The following axiom is the syntactic counterpart of the meta-flow axiom,
modified slightly to be laxer about what it considers to be a solution
(namely allowing data flow variable to be satisfied, without requiring 
that the corresponding synchronisation variable are defined).

\begin{definition}[Simple Flow Axiom]
	\begin{equation}
	SFA(x) ~~\defn~~\true ~\slor~ x ~\slor~ \neg x ~\slor~ (x \wedge
	\wh{x}=\wh{x}) ~\slor~ \wh{x}=\wh{x}
	\tag{simple flow axiom}
	\end{equation}
\end{definition}

We write $SFA(X)$ for the conjunction $\bigwedge_{x\in X}SFA(x)$. We also write $SFA(\psi)$ as a shorthand for $SFA(\fv{\psi})$.

\begin{lemma}\label{lemma:SFAexpressesMFA}
	$\sigma,\I\smodels SFA(x)$ if and only if $\sigma^P$ satisfies the meta-flow axiom, where $\sigma^P$ extends $\sigma$ as follows:
	\[\sigma^P = \sigma \cup \set{x\mapsto \True ~|~ \wh{x} \in \dom(\sigma)}\]
\end{lemma}
\begin{proof}
	We have $\emptyset,\I\smodels\true$;~
	$\set{x\mapsto\true},\I\smodels x$;~ 
	$\set{x\mapsto\false},\I\smodels \neg x$;
	$\set{x\mapsto\true,\wh{x}\mapsto t},\I\smodels x \wedge \wh{x}=\wh{x}$; and
	$\set{\wh{x}\mapsto t},\I\smodels \wh{x}=\wh{x}$, for
	an arbitrary ground term $t$, and no other $\sigma$.
	Extending $\sigma$ to $\sigma^P$ we obtain precisely the solutions to the  
	meta-flow axiom.
\end{proof}

Observe that simple logic clearly does not preserve classical or even partial 
equivalences, as $\true \slor x \slor$ $\lnot x \slor (x \wedge \wh{x}=\wh{x}) 
\slor \wh{x}=\wh{x}\equiv_C\true$,
classically, but this is not the case in simple logic.

\begin{definition}[Simple Firing] \label{def:simplefiring}
	~\\An assignment $\sigma$ is called a \emph{simple firing} whenever  
	$\sigma,\I\smodels \psi\sland SFA(\psi)$.
\end{definition}

Note that the simple flow axiom differs from the meta-flow axiom
because it also accepts solutions where $\wh{x}$ is defined, but $x$ is not.
This is because the simple flow axiom is designed to filter from a set of
minimal solutions (i.e., solutions in the simple logic), while the meta-flow
axiom is designed to filter good solutions from all  solutions the constraint
engine finds, namely, the ones that include additional assignments to make the flow axiom hold. 
As a consequence, the assignment $\set{\wh{a} \mapsto d}$ is a simple firing for
the formula $\wh{a} = d$, and the assignment $\set{\wh{a} \mapsto d,a\mapsto
\True}$ is a partial firing for the same formula, but not the other way
around.

With simple logic we can check a formula to ensure that all of its solutions satisfy
the meta-flow axiom. This means that we do not need to filter solutions to such a
constraint. Furthermore, as the simple flow axiom is preserved through composition
($\sland$), we are guaranteed to have simple firings without having to perform
a post hoc filter phase.

Note that implication in simple logic does not have the exact same meaning as
in the other logics. $c\to a$, whenever in a positive position, is regarded in
simple logic as $\lnot(c\sland \lnot a)$, which has only two firings:
$\set{c\mapsto\false}$ and $\set{a\mapsto\true}$. The union of these two firings is not a simple firing because it is not ``minimal enough". That is, the resulting union is not satisfied by the simple satisfaction relation because it contains too many elements.
Recall the constraints of the \filterp in \autoref{tab:statelessex}, and let $d$ be such that $p(d)$ does not hold. The assignment
$\set{a\mapsto\true,c\mapsto\false,\wh{a}\mapsto d}$ is both a partial and a simple firing. However,
$\set{z\mapsto\true,a\mapsto\true,c\mapsto\false,\wh{a}\mapsto d}$
is also a partial firing but not a simple firing, since 
$z \notin \fv{c\to a}$, therefore the firing is not mininal enough.

\begin{lemma} \label{lem:scomposition}
Let $\psi_1$ and $\psi_2$ be constraints defined for the simple logic. Then
\[
 (\psi_1\sland SFA(\psi_1)) ~\land~  (\psi_2\sland SFA(\psi_2))
 ~~~\equiv~~~
 (\psi_1\land\psi_2)  ~\sland~  (SFA(\psi_1 \land \psi_2))
 \]
The equivalence between the left and the right hand formul\ae\ mean that they
have the same solutions according to the simple satisfaction.
\end{lemma}

\begin{proof}
Let $\mathsf{sols}(\psi)$ denote the set of solutions of $\psi$ according to
the simple satisfaction relation, and let  $\mathsf{sols}(\psi_1) = S_1$,
$\mathsf{sols}(\psi_2) = S_2$, $\mathsf{sols}(SFA(\psi_1)) = S_{F1}$, and
$\mathsf{sols}(SFA(\psi_2)) = S_{F2}$. The proof follows from the expansion of
the definition of simple satisfaction.
\[\begin{array}{rcl}
  \multicolumn{3}{l}{
  \mathsf{sols}((\psi_1\sland SFA(\psi_1)) ~\land~
                (\psi_2\sland SFA(\psi_2)))}\\
   & = &
    \set{\sigma_1 \cup \sigma_2 ~|~
      \sigma_1 \in S_1 \cap S_{F1}, \sigma_2 \in S_2 \cap S_{F2},
      \sigma_1 \frown \sigma_2}\\
   & = &
    \set{\sigma_1 \cup \sigma_2 ~|~
      \sigma_1 \in S_1, \sigma_1 \in S_{F1}, \sigma_2 \in S_2,
      \sigma_2 \in S_{F2}, \sigma_1 \frown \sigma_2}\\
   & = &
    \set{\sigma_1 \cup \sigma_2 ~|~
      \sigma_1 \in S_1, \sigma_2 \in S_2, \sigma_1 \frown \sigma_2}
   ~\cap~
    \set{\sigma_1 \cup \sigma_2 ~|~
      \sigma_1 \in S_{F1}, \sigma_2 \in S_{F2}, \sigma_1 \frown \sigma_2}\\
  & = &  
    \mathsf{sols}(\psi_1\land\psi_2) ~\cap~
    \mathsf{sols}(SFA(\psi_1) \land SFA(\psi_2))\\
  & = &  
    \mathsf{sols}((\psi_1\land\psi_2) ~\sland~ 
                  (SFA(\psi_1) \land SFA(\psi_2)))\\
  & = &  
    \mathsf{sols}((\psi_1\land\psi_2) ~\sland~ (SFA(\psi_1 \land \psi_2)))
\end{array}\]
\end{proof}

\begin{lemma}[Partial to Simple]
	\ \begin{itemize}
		\item 
	If $\sigma,\I\pmodels\psi$ and $MFA(\sigma)$, then there exists
	$\sigma^\ddag$ such that $\sigma^\ddag\subseteq\sigma$ and 
	$\sigma^\ddag,\I\smodels\psi^S \sland SFA(\psi)$.
	\item
	If $\sigma,\I\copmodels\psi$ and $MFA(\sigma)$, then there exists
	$\sigma^\ddag$ such that $\sigma^\ddag\subseteq\sigma$ and 
	$\sigma^\ddag,\I\cosmodels\psi^S \sland SFA(\psi)$.
\end{itemize}
\end{lemma}
\begin{proof}
	Proof is by straightforward induction on $\psi$. Note that $\sland$ cannot 
	occur in $\psi$, and in each step $\sigma^\ddag$ is guaranteed to exist and 
	to obey the simple flow axiom:
	\begin{description}
		\item[Case $\true$] --- $\sigma^\ddag=\emptyset$.
		\item[Case $x$] --- $\sigma^\ddag=\set{x\mapsto\sigma(x)}$.
		\item[Case $\psi_1\wedge\psi_2$] --- For the $\pmodels$ case: 
			Assume that $\sigma,\I\pmodels \psi_1\wedge\psi_2$ and $MFA(\sigma)$.
			Therefore $\sigma,\I\pmodels \psi_1$ and $\sigma,\I\pmodels \psi_2$. By
			the induction hypothesis, we have $\sigma^\ddag_1\subseteq\sigma_1$ and
			$\sigma^\ddag_2\subseteq\sigma_2$ such that $\sigma^\ddag_1,\I\smodels
			\psi_1^S \sland SFA(\psi_1)$ and $\sigma^\ddag_2,\I\smodels
			\psi_2^S \sland SFA(\psi_2)$. Clearly we have
			$\sigma^\ddag_1\frown\sigma^\ddag_2$ 
			and
			$\sigma^\ddag_1\cup\sigma^\ddag_2\subseteq\sigma$,
			and by Lemma~\ref{lem:scomposition} we conclude that
			$\sigma^\ddag_1\cup\sigma^\ddag_2,\I\smodels (\psi_1\wedge\psi_2)^S
			\sland SFA(\psi_1\land\psi_2)$.
		
  		For the $\copmodels$ case:
			Assume that $\sigma,\I\copmodels \psi_1\wedge\psi_2$ and $MFA(\sigma)$.
			Therefore $\sigma,\I\pmodels \psi_1$ or $\sigma,\I\pmodels \psi_2$. By
			the induction hypothesis, we have $\sigma^\ddag_1\subseteq\sigma_1$ and
			$\sigma^\ddag_2\subseteq\sigma_2$ such that $\sigma^\ddag_1,\I\smodels
			\psi_1^S \sland SFA(\psi_1)$ and $\sigma^\ddag_2,\I\smodels \psi_2^S \sland SFA(\psi_2)$. Note that $\land$ is in a negative position,
			therefore $(\psi_1 \land \psi_2)^S = \psi_1 \sland \psi_2$. Clearly we
			have that when $\sigma^{\ddag}=\sigma^{\ddag}_1$ or
			$\sigma^{\ddag}=\sigma^{\ddag}_2$, $\sigma^{\ddag},\I\cosmodels
			(\psi_1\wedge\psi_2)^S \sland SFA(\psi_1\land\psi_2)$ and
			$\sigma^\ddag\subseteq\sigma$.

					\item[Case $\neg\psi$]	--- by induction hypothesis.	
		\item[Case $p(t_1,\ldots,t_n)$]	--- in both cases 
			$\sigma^\ddag=\set{ v\mapsto\sigma(v) ~|~ v\in \fv{p(t_1,\ldots,t_n)}}$
	\end{description}
\end{proof}

\begin{lemma}\label{lemma:temporary}
	If $\sigma,\I\smodels\psi$, then $\sigma^P,\I\pmodels\psi^P$.
	If $\sigma,\I\cosmodels\psi$, then $\sigma^P,\I\copmodels\psi^P$.
\end{lemma}
\begin{proof}
By straightforward induction on $\psi$.
\end{proof}

\begin{lemma}[Simple to Partial]
	If $\sigma$ is a simple firing for $\psi$, then $\sigma^P$ is a partial firing for $\psi^P$.
	Furthermore, for all $\sigma'$ such that $\sigma^P\subseteq\sigma'$ 
	and $\sigma'$ satisfies the meta-flow axiom, $\sigma'$ is a partial firing for $\psi^P$.
\end{lemma}
\begin{proof}
	Follows from Lemmas~\ref{lemma:subsetpartial} and~\ref{lemma:temporary}.
\end{proof}

\

The key difference between simple and partial is that simple finds the kernel
of a solution by examining only the relevant variables. All partial solutions
can be reconstructed by filling in arbitrary values 
(satisfying the meta-flow axiom) for the unspecified variables.
Note that the classical model is faithful to existing semantics of \reo. By
shifting to a partial logic, we can model \emph{pure synchronisation}, which 
is when a synchronisation variable is true and the corresponding data flow variable
is $\bot$. In the simple logic, if the data flow variable is not mentioned,
it will never be assigned a value, reflecting that there is no data
flowing in the corresponding ports, i.e., it is a pure synchronisation port.

\section{Locality}
\label{sec:locality}

With simple logic, there is still a single set of constraints and thus it is
not clear how to exploit this to extract any inherent concurrency nor is it
clear how to partition the constraints to distribute them. Our motivation
is to use (distributed) constraint satisfaction as the basis of a coordination language
between geographically distributed components and services.

The local semantics is based on a \emph{configuration} consisting of
constraints partitioned into blocks of
constraints, 
denoted by $\Psi = \block{\psi_1}{},\ldots,\block{\psi_n}{}$, or
simply $\Psi = \block{\psi_i}{}^{i\in 1..n}$.

\begin{definition}[No-Flow Assignment]
	An assignment $\sigma$ is called a \emph{no-flow assignment} whenever
	$\dom(\sigma)$ $\subseteq\dvar{X}$ and for all $x\in\dom(\sigma)$ we have
	$\sigma(x)=\false$.
\end{definition}

\begin{axiom}[No-Flow Axiom]
	We say that a constraint $\psi$ obeys the no-flow axiom whenever
	there is some no-flow assingment $\sigma$ with $\dom(\sigma)\subseteq \fv{\psi}\cap\dvar{X}$
	such that $\sigma,\I\smodels\psi$.
	
	A configuration $\Psi = \block{\psi_i}{}^{i\in 1..n}$  obeys the no-flow axiom $\ifff$ 
	 each $\psi_i$ obeys the no-flow axiom.
\end{axiom}

From now on, we assume that all configurations satisfy the no-flow axiom.

\begin{definition}[Boundary]
	Given a configuration $\Psi = \block{\psi_1}{},\ldots,\block{\psi_n}{}$,
	define $\borderr{\Psi}{\block{\psi_i}{}}$ as 
	 $\fv{\psi_i}\cap ~$ $\fv{\Psi_{-i}}$,
	where $\Psi_{-i} = \block{\psi_1}{},\ldots,\ldots,\block{\psi_{i-1}}{},\block{\psi_{i+1}}{},\ldots,\block{\psi_n}{}$.
	
	We drop the $\Psi$ subscript from $\borderr{\Psi}{-}$ when it is clear from the context.
\end{definition}

\begin{definition}[Local Firing]
\label{def:localfiring}
	Given a configuration $\Psi = \block{\psi_1}{},\ldots,\block{\psi_n}{}$. We say that:
	\begin{itemize}
  \item	$\sigma$ is a \emph{local firing} for a block $\block{\psi_i}{}$ if 
    and only if $\sigma$ is a simple firing for $\psi_i$ \emph{and}
    for all $x\in\borderr{\Psi}{\block{\psi_i}{}}$ we have $\sigma^P(x)\neq\true$---we call
this the \emph{boundary condition}.\footnote{$\sigma^P$ is defined in 
\autoref{lemma:SFAexpressesMFA}.}
	\item $\sigma$ is a \emph{local firing} for $\Psi$ if and only if
	  $\sigma=\sigma_1\cup\cdots\cup\sigma_{m'}$ such that 
    \begin{enumerate}
      \item $I_1,\ldots,I_{m'},\ldots,I_m$ is a partition of $\set{1..n}$;
      \item $\varphi_i = \bigwedge_{j\in I_i}\psi_j$ where $i\in 1..m$;  and
      \item $\sigma_i$ is a local firing for block $\block{\varphi_i}{}$ where $i\in1..m'$.
    \end{enumerate}
  \end{itemize}
  \vspace{-7mm}
\end{definition}

\noindent
The intuition behind this definition is:
\begin{enumerate}
	\item a local firing can occur in some isolated block \emph{or} the
	conjunction of some blocks \emph{or} in independent (conjunctions of) blocks; and
	\item within each block a simple firing occurs that makes the assumption
		that there is no-flow on its boundary ports.
\end{enumerate}

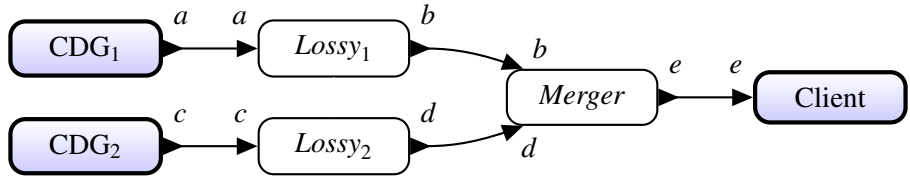
\begin{figure}[htb]
\centering
\begin{tikzpicture}
  [node distance=3.3cm,
      >=triangle 45,thick, bend angle=10,
   dot/.style={circle,inner sep=0.5mm,fill=black},
   closer/.style={node distance=1.3cm}
  ]    
  \node[main,label=10:$a$]                      (cdg1) {CDG$_1$};
  \node[prim,label=170:$a$,label=10:$b$,right of=cdg1] (lossy1) {$\lossy_1$};
  \node[main,closer,label=10:$c$,below of=cdg1] (cdg2) {CDG$_2$};
  \node[prim,closer,label=170:$c$,label=10:$d$,below of=lossy1]
    (lossy2) {$\lossy_2$};
  \draw[-,color=white] (lossy1) to node (m) {} (lossy2);
  \node[prim,label=130:$b$,label=220:$d$,label=10:$e$,right of=m]
    (merger) {\textit{Merger}};
  \node[main,label=170:$e$,right of=merger] (client) {Client};
  
  \draw[>->] (cdg1) to (lossy1);
  \draw[>->] (cdg2) to (lossy2);
  \draw[>->,bend left]  (lossy1.east) to (merger);
  \draw[>->,bend right] (lossy2.east) to (merger);
  \draw[>->] (merger) to (client);

\end{tikzpicture}

\caption{Simple network of constraints: two competing data producers.}
\label{fig:lossymerge}
\end{figure}

\begin{example}
\label{ex:lossymerge}
We introduce a small example in \autoref{fig:lossymerge} that we use to
illustrate the definition of local firings. Let $\phi_i$, where $i\in 1..6$,
be the constraints for CDG$_1$, CDG$_2$, $\lossy_1$, $\lossy_2$,
$\mathit{Merger}$, and Client, respectively. We define $\phi_{3}$ for
$\lossy_1$ and
$\phi_{5}$ for $\mathit{Merger}$ as follows. 
\[\begin{array}{lcl}
  \phi_{3} &=& b\to a ~~\land~~ b\to(\wh{a}=\wh{b})
  \\
  \phi_{5} &=& e \leftrightarrow (b \lor d) ~~\land~~ \lnot(b\land d)
            ~~\land~~ b\to (\wh{e}=\wh{b})
            ~~\land~~ d\to (\wh{e}=\wh{d})
  \end{array}
\]
The remaining constraints can be
derived similarly.

The $\lossy_1$ can arbitrary lose data flowing on $a$, or pass data from $a$ to
$b$; and $\mathit{Merger}$ can pass data either from $b$ to $e$ or from $d$ to
$e$. The configuration that captures the behaviour of the full system is given
by $\Psi_{merge} = \block{\phi_i \sland SFA(\phi_i)}{}^{i\in 1..6}$.
\end{example}

We present some simple firings for the \lossy primitive, and show which of these are also local firings for $\Psi_{merge}$, and we then describe
some more complex local firings for $\Psi_{merge}$.
We omit the formal proof that the conditions for local firings hold for these firings.
Formula $\phi_{3}$ can be written as $\lnot(b \sland \lnot a) \land \lnot(b
\sland \lnot(\wh{a}=\wh{b}))$, and the boundary of $\block{\phi_{3}}{}$ is
$\set{a,\wh{a},b,\wh{b}}$. Valid simple firings for
$\phi_{3}$ are $\set{b\mapsto\false}$,
$\set{a\mapsto\true,b\mapsto\false}$, and $\set{a\mapsto\true,
\wh{a}\mapsto v, \wh{b}\mapsto v}$, for any possible data value $v\in\Data$.
The only simple firing that is also a local firing
is then $\set{b\mapsto\false}$. This means that $\set{b\mapsto\false}$ is also
a local firing of $\Psi_{merge}$.

It is also possible to show that the
solution $\sigma_{\it top}$ corresponding to the flow of some data value $v$
from CDG$_1$ to
Client is a simple firing for $\phi_1\land\phi_3\land\phi_{5}\land\phi_{6}$,
and that the solution $\sigma_{\it bottom}$ corresponding to data being sent
from CDG$_2$ to $\lossy_2$ and being lost is a simple firing for
$\phi_2\land\phi_{4}$. More precisely, we define $\sigma_{\it top} = \set{a\mapsto\true, b\mapsto\true, d\mapsto\false, e\mapsto\true, \wh{a}\mapsto v, \wh{b}\mapsto v, \wh{e}\mapsto v}$ and $\sigma_{\it bottom}=\set{c\mapsto\false, d\mapsto\false}$.
The boundary for both sets of primitives is just 
$\set{d}$. In the solutions $\sigma_{\it top}$ and $\sigma_{\it bottom}$ the 
value of $d$ is never $\true$,
so the boundary conditions hold.
Therefore $\sigma_{\it top}$, $\sigma_{\it
bottom}$, and $\sigma_{\it top}\cup\sigma_{\it
bottom}$ are also local firings of $\Psi_{merge}$.

The local semantics is based on the simple semantics under the no-flow axiom 
assumption.
Thus, a
 simple firing can be trivially seen as a local solution, but a local firing
needs to be extended to be seen as a simple firing. This extension corresponds
exactly to the unfolding of the no-flow axiom for the blocks of constraints
not involved in the local firing. The embedding between these two semantics is 
formalised below.

\begin{lemma}[Simple to Local]
	If $\sigma$ is a simple firing for $\psi$, then 
	$\sigma$ is a local firing for $\Psi=\block{\psi}{}$.
\end{lemma}
\begin{proof}
	As $\borderr{\Psi}{\block{\psi}{}}=\emptyset$, 
	a simple firing for $\psi$ is also a local firing for $\psi$.
\end{proof}

\begin{lemma}[Local to Simple]
  Let $\sigma$ be a local firing for 
  $\Psi=\block{\psi_1}{},\ldots,\block{\psi_n}{}$, 
  then there exists a $\sigma^\star$ such that
(1)~$\sigma\subseteq\sigma^\star$, (2)~for all $x\in\dom(\sigma^\star)\setminus\dom(\sigma)$
we have $\sigma^\star(x)=\false$, and (3)~$\sigma^\star$ is a simple firing for
$\bigwedge_{i\in1..n}\psi_i$.
\end{lemma}

\begin{proof}
	Assume that $\sigma$ is a local firing for $\Psi=\block{\psi_1}{},\ldots,\block{\psi_n}{}$.
	Without loss of generality, we can assume that $\sigma=\sigma_1\cup\cdots\cup\sigma_m$
	where (1)~$m\le n$ and for each $k\in 1..m$ we have that $\sigma_k$ is a local
	firing for $\block{\psi_k}{}$. From the no-flow axiom, we can have a no-flow assignment
	$\sigma^\star_j$ for each $j\in m+1..n$ such that $\sigma^\star_j,\I\smodels\psi_j$.
	From the boundary condition, we can infer that for each $\sigma_k$ and $\sigma^\star_j$
	the condition $\sigma_k\frown\sigma^\star_j$ holds.  Thus, for
	 $\sigma^\star=\sigma\cup\bigcup_j\sigma^\star_j$,
	we have $\sigma^\star,\I\smodels\bigwedge_{i\in1..n}\psi_i$. 
\end{proof}

\section{State}
\label{sec:state}

\subsection{State-machine}
\label{sec:statemachine}

We follow the encoding of stateful connectors presented by Clarke et 
al.~\cite{reo:deconstructing}.
To encode stateful connectors we add $\mathit{state}_p$ and
$\mathit{state'}_p$ to the term variables, for each $p \in P$ 
corresponding to a stateful primitive. A state machine with states
states $q_1, \ldots, q_n$  is encoded as a formula of the form:
\[ \psi = \mathit{state}_p = q_1 \to \psi_1 \land \ldots \land
          \mathit{state}_p = q_n \to \psi_n
\]
where $\psi_1,\ldots,\psi_n$ are 
constraints representing the transitions from each state.
For each firing $\sigma$, the value of $\sigma(\mathit{state}'_p)$ determines 
how the connector evolves, giving the value of the next state.

We illustrate this encoding by an example, presenting 
the constraints encoding the state machine of a \fifo buffer
from \autoref{tab:statelessex}.
The state can be either $\mathtt{empty}$ or
$\mathtt{full}(d)$, where $d \in \Data$, and $\mathtt{empty}$ and 
$\mathtt{full}$ are function symbols in $\dvar{F}$. We then define
$\mathit{fifo}$-constraint to be $\mathit{state}_{\it fifo} =
\mathtt{empty} \to \psi_e \land \mathit{state}_{\it fifo} = \mathtt{full}(d)
\to \psi_f$, where $\psi_e$ and $\psi_f$ are the upper and lower labels of the
following diagram, respectively:
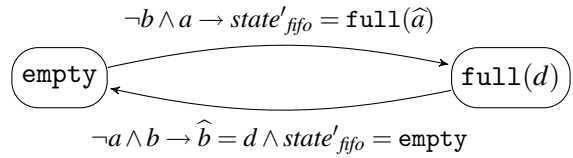
\[\begin{tikzpicture}  [shorten >=1pt,>=stealth',node distance=6cm,auto,initial text=,
   bend angle=12,
   eliptic/.style={rectangle,rounded corners=3.5mm,draw=black,
   fill=white, minimum height=8mm}]
    \node[eliptic] (e) {$\mathtt{empty}$};
  \node[eliptic,right of=e] (f) {$\mathtt{full}(d)$};
  \path[->] (e) edge [bend left] node
            {\small $\lnot b \land a\to \mathit{state'}_{\it
              fifo}=\mathtt{full}(\wh{a})$} (f)
            (f) edge [bend left] node
            {\small $\lnot a \land b\to \wh{b}=d \land
            \mathit{state'}_{\it fifo}=\mathtt{empty}$} (e);
\end{tikzpicture}
\]
To complete the encoding, we add a formula describing the present state to the
mix. In the example, the formula $\mathit{state}_{\it fifo}=\mathtt{empty}$
records the fact that the \fifo is in the empty state, whereas
$\mathit{state}_{\it fifo}=\mathtt{full}(d)$ records that it is in the full
state, containing data $d$. The full constraint for the \fifo primitive is now
(refining the constraints in Table~\ref{tab:statelessex}):
\[
\mathit{state}_{\it fifo} =\mathtt{empty} \to \psi_e ~~\land ~~\mathit{state}_{\it fifo} =
\mathtt{full}(d) \to \psi_f ~~\land~~ \mathit{state}_{\it fifo} = empty,
\]

\subsection{Constraint satisfaction-based engine}
\label{sec:engine}

A constraint satisfaction-based engine holds a configuration with the current set of constraints and operates in \emph{rounds}, each of which
consists of a \emph{solve} phase and an \emph{update} phase, which uses the
firing to update the constraints and to model the transition to a new state.
This is depicted in \autoref{fig:phases}.

\begin{figure}[htb]
\centering
\begin{tikzpicture}
  [node distance=5.5cm,
   text height=3.5ex,text depth=2.5ex,
   >=triangle 45, bend angle=20,
   place/.style={rectangle,minimum width=30mm,rounded corners=5mm,
       inner sep=0.5mm,thick,draw=black,top color=white,bottom color=black!20}
  ]    
  \node[place] (A)              
    {\wrap{ Configuration\\
     $\tpl{\rho_i\land\epsilon_i}^{i \in 1..n}$}};
  \node[place] (B) [right of=A]
    {\wrap{ Firing\\$\sigma$}}
    edge [<-,bend right] node[fill=white] {\footnotesize Solve}  (A)
    edge [->,bend left]  node[fill=white] {\footnotesize Update} (A);
\end{tikzpicture}

\caption{Phases of the constraint satisfaction-based engine.}
\label{fig:phases}
\end{figure}
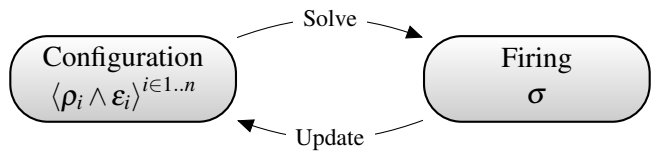

Each block of the configuration is a conjunction of two constraints $\tpl{\rho \land \epsilon}$, where $\rho$
is persistent and $\epsilon$ is ephemeral. Persistent constraints are
eternally true, and can be either (normal) stateless constraints, stateful
constraints, or the conjunction of persistent constraints. Ephemeral
constraints describe the present state of the stateful constraints.
Configurations are updated at each round.
Let \I be an interpretation and, for each $i \in 1..n$, let $P_i$ be a (possibly 
empty) set of names of stateful constraints.
A full round can be
represented as follows, where the superscript indicates the round number:
\begin{equation*}
	\tpl{\rho_i\land\epsilon_i^m}^{i\in 1..n} \xrightarrow{solve} \tpl{\sigma^m}
	\xrightarrow{update} \tpl{\rho_i\land\epsilon_i^{m+1}}^{i\in 1..n} 
\end{equation*}
satisfying the following:
\begin{align}
& \sigma^m,\I\smodels \bigwedge_{i \in 1..n}\rho_i\wedge\epsilon_i^m   \tag{solve}\\
& \begin{array}{@{}lcl@{}}
    \epsilon_i^{m+1} &\equiv& \bigwedge \set{
    state_p = \sigma^m(\mathit{state'}_p)
    ~|~
    p \in P_i \mbox{ and }
    \sigma^m(\mathit{state'}_p) \neq \bot} ~\land
    \\ 
    &&
    \bigwedge \set{ \epsilon_i^m ~|~
    p \in P_i \mbox{ and }
    \sigma^m(\mathit{state'}_p) = \bot}
  \end{array}
  \tag{update}
\label{eq:nexteps}
\end{align}

In round $m$, the solve phase finds a solution $\sigma^m$, and the
update phase replaces the definition of $\epsilon^m$ by $\epsilon^{m+1}$ for
round $m$, whenever the variable $\mathit{state'}$ is defined.
A correctness result of this approach with respect to \reo, for the classical
semantics, has been presented by Clarke et al.~\cite{reo:deconstructing}. The
authors use the constraint automata semantics for \reo~\cite{reo:ca06} as the
reference for comparison. The present approach adapts the previous one
by accounting for partial solutions of the constraints, which means that
only some of the state variables are updated.

\section{Interaction}
\label{sec:interaction}

We now extend the model with means for external interaction.

\subsection{External functions, predicates and constraints}
\label{sec:externallogic}

We defined in \autoref{sec:classical} a core syntax for logic formul\ae, extended with state variables in \autoref{sec:state}. Satisfaction of formul\ae\ is defined with respect to an assignment $\sigma$ defining the values of variables, and an interpretation \I giving meaning to predicates.
We now extend the syntax of the logic and the definition of interpretation
\I, by introducing new symbols whose interpretation is also given by \I.
These symbols are \emph{external predicate} symbols
$\evar{p} \in \edom{P}$, \emph{external function} symbols $\evar{f} \in
\edom{F}$, and  \emph{external constraints}  $\evar{c} \in \edom{C}$. We also introduce
\emph{communication variables} $\evar{k} \in \edom{K}$ whose value in the
solution of a round can be communicated to the outside world. Formul\ae\ are
now given by the following syntax:
\begin{eqnarray*}
	\psi & ::= &
    \true ~~|~~ x ~~|~~ \psi_1 \land \psi_2 ~~|~~ \lnot \psi ~~|~~ 
    p(\lotsof{t}) ~~|~~ \evar{p}(\lotsof{t}) ~~|~~ 
    \evar{c}(\lotsof{\psi},\lotsof{t})\\
 	t & ::= &
 	  \wh{x} ~~|~~ \mathit{state}_p ~~|~~ \mathit{state'}_p ~~|~~
 	  \evar{k} ~~|~~ f(\lotsof{t}) ~~|~~ \evar{f}(\lotsof{t})
\end{eqnarray*}

Use $\lotsof{t}$ as a shorthand for $t_1,\ldots,t_n$.
We extend the definition of interpretation to be 
an arity-indexed family of \emph{partial} map from $P_n\times
\T^{n}$ to $\set{\true,\false}$; from $\evar{P}_n\times \T^{n}$ to $\set{\true,
\false}$; from
$\evar{F}_n\times \T^{n}$ to ground terms; and from $\evar{C}$ to a 
term with $l$  formul\ae\ parameters and $k$ term parameters. 

We also extend the $\mathit{Val}_{\sigma,\I}$ function, which is now
parameterized on $\sigma$ and $\I$. This function replaces variables $v$ by
$\sigma(v)$ and $\evar{f}(t_1,\ldots,t_n)$ by
$\I(\evar{f},\mathit{Val}_{\sigma,\I}, \ldots, \mathit{Val}_{\sigma,\I})$,
or is undefined if any component is undefined.

The extension of the syntax of the logic requires the addition of two new
(dis)satisfaction rules:
\[
\begin{array}{ll}
  \sigma,\I\smodels \evar{p}(t_1,\ldots,t_n) &\ifff\quad         \evar{p}(\mathit{Val}_{\sigma,\dvar{I}}(t_1), \ldots,\mathit{Val}_{\sigma,\dvar{I}}(t_n))\mapsto\true\in\I
        \\&~~~~\quad\mbox{and }
        dom(\sigma) = \mathit{fv}(p(t_1\ldots,t_n))
  \\     
  \sigma,\I\smodels \evar{c}(\psi_1,\ldots,\psi_m,t_1,\ldots,t_n) \quad 
  &\ifff\quad
    \sigma,\I \smodels
    \psi[\psi_1/v_1,\ldots,\psi_m/v_m,
    t_1/v_{m+1},\ldots,t_m/v_{m+n}]
        \\&
    \textrm{where }\evar{c} \mapsto \lambda(v_1,\ldots,v_{m+n}).\psi \in \I \\
\\
\sigma,\I\cosmodels \evar{p}(t_1,\ldots,t_n) &\ifff\quad       \evar{p}(\mathit{Val}_{\sigma,\dvar{I}}(t_1), \ldots,\mathit{Val}_{\sigma,\dvar{I}}(t_n))\mapsto\false\in\I
      \\&~~~~\quad\mbox{and }
      dom(\sigma) = \mathit{fv}(p(t_1\ldots,t_n))
 \\     
\sigma,\I\cosmodels \evar{c}(\psi_1,\ldots,\psi_m,t_1,\ldots,t_n) \quad 
&\ifff\quad
  \sigma,\I \cosmodels
  \psi[\psi_1/v_1,\ldots,\psi_m/v_m,
  t_1/v_{m+1},\ldots,t_m/v_{m+n}]
    \\&
  \textrm{where }\evar{c} \mapsto \lambda(v_1,\ldots,v_{m+n}).\psi \in \I
\end{array}
\]
The notation $\lambda(v_1,\ldots, v_n).\psi$ denotes that $\psi$ is a formula
where $\set{v_1,\ldots, v_n} \subseteq \fv{\psi}$.  Each $v_i$ is a variable
that acts as a placeholder for $\psi$, that is substituted when evaluating the
external variable mapped to the $\lambda$-term, hence the $\lambda$-notation.


\subsection{External world}

The constraint-based engine introduced in \autoref{sec:engine} describes the
evolution of a configuration (a set of blocks of constraints). We now assume
the existence of a set of primitives $P$, each of which provides a single
block of constraints to  the engine. These primitives can be one of
three kinds~\cite{reo:deconstructing}:
\begin{description}
 \item[internal, stateless] denoted by $P_{no}$. The underlying constraints
 involve neither state variables nor communication variables in $\edom{K}$, and all
 constraints are persistent---represented by setting the ephemeral constraints
 to $\epsilon_p = \true$, where $p\in P_{no}$.

 \item[internal, stateful] denoted by $P_{int}$. Such primitives have
 constraints over the state variable pair $\mathit{state}_p$ and
 $\mathit{state'}_p$, where $\mathit{state}_p$ represents the value of the
 current state of $p \in P_{int}$ and $\mathit{state'}_p$ the value of the
 next state. The emphemeral constraint denotes the current state and is always
 of form $\epsilon_p\equiv \mathit{state}_p = t$, for some ground term $t$. No
 communication variables may appear in the constraints.

 \item[external] denoted by $P_{ext}$. Such primitives express constraints in
 terms of a communication variable $\evar{k}$ through which data is passed
 from a primitive $p\in P_{ext}$ to the outside world. The outside world then
 sends a  new set of constraints to represent $p$'s next step behaviour.
 No state variables can appear in the constraints, as it is assumed that the
 state information is handled externally and incorporated into the constraints
 sent during the update phase.
\end{description}

We assume that the constraints $\psi_p$ provided by each primitive $p\in P$ can only 
have a fixed set of free variables, denoted by $\fv{p}$.
Note that $\fv{\psi_p} \subseteq \fv{p}$.
The relation between external symbols, communication variables and the
external primitives in $P_{ext}$ is made via an \emph{ownership} relation. 
That is, each external symbol and each communication variable is \emph{owned}
by a unique primitive in $P_{ext}$.

\begin{definition}[Ownership]
 Let $\dvar{O} =
 \edom{F}\cup\edom{P}\cup\edom{C}\cup\edom{K}$. Each $o \in \dvar{O}$ is
 managed by exactly one $p \in P_{ext}$. This is denoted using 
 function $\mathit{own}: \dvar{O} \to  P_{ext}$. 
We may write  $\mathbf{k}_p$ to indicate that $\mathit{own}(\evar{k})=p$.

 We write  $\block{\psi}{Q}$ to indicate that the constraints in $\psi$
 are owned by primitives $Q$, where $Q\subseteq P$.
\end{definition}

\begin{example}
We extend the constraints of our running example from \autoref{tab:statelessex}, presented in \autoref{tab:interactiveex}.
Using the updated constraints from \autoref{tab:interactiveex}, the global
constraint is given by the configuration $\block{\psi_i}{}^{i \in 1..7}$, the synchronous variables are $\dvar{X} = \set{a,b,c}$, the only
uninterpreted predicate symbol is equality, $\mathbf{more} \in \edom{C}$ is an
external constraint symbol, $\mathbf{result}\in\edom{K}$ is a communication
variable, and $\mathbf{UserAppr} \in \edom{P}$ is an external predicate
symbol. Furthermore, $own(\evar{more})=\textrm{CDG}$, $own(\mathbf{result})=\textrm{Client}$, and $own(\evar{UserAppr})=\textrm{User approval}$.

\begin{table}[htb]
{\centering
$\begin{array}{clc}
\toprule
~\hspace{15mm} Primitive \hspace{15mm}~ & \multicolumn{2}{c}{Constraint} \\
\cmidrule(lr){1-1}\cmidrule(lr){2-3}
\myprim{
  \node at (-1.5,0) {~};
  \node[main] (x) at (0,0) {CDG};
  \node at (1.5,0) {a} edge[<-] (x);}
& \psi_1 = &     {a \to (a \land (\widehat{a} = \mathtt{d_1} \lor \widehat{a} =
    \mathtt{d_2} \lor \widehat{a} = \mathtt{d_3} \lor     \mathbf{more}(\wh{a})))}
\\
\myprim{
  \node at (-1.5,0) {c} edge [->] (x);
  \node[main] (x) at (0,0) {Client};
  \node at (1.5,0) {~};}
& \psi_2 = &
    { \mathbf{result} = \widehat{b} }
\\
\myprim{
  \node[prim] (x) at (0,0) {User approval};
  \node at (-1.8,0) {c} edge[->] (x);
  \node at (1.8,0) {~};}
& \psi_3 = &     {c \to (c \land \mathbf{UserAppr}(\widehat{c}))}
\\
\bottomrule
\end{array}$

}
\caption{Updated (interactive) constraints of a set of primitives.}
\label{tab:interactiveex}
\end{table}
\end{example}

The updated constraints in \autoref{tab:interactiveex} illustrate the usage of
the extensions to the logic. External constraints can model \emph{on-the-fly}
constraint generation. The interpretation of $\evar{more}$ can refer to new
external constraints, and this process can be repeated an unbounded number of
times. Communication variables provide a mean to communicate the result to the
external world, as the constraints of Client show, via the variable
$\textbf{result}$. Finally, the external
predicate $\evar{UserAppr}$ in $\psi_3$ illustrates the possibility of asking
external primitives if some predicates hold for specific data values.

\paragraph{Example of the execution of the engine}

Recall Example~\autoref{ex:lossymerge}, which is based in a set of primitives
$P$. We partition $P$ into $Q$ and $R$, where $Q=\set{\textrm{Client},
\textrm{CDG}_1, \lossy_1, \mathit{Merger}}$, and $R=\set{\textrm{CDG}_2,
\lossy_2}$. To provide a better understanding of how the engine  evolves
with respect to the external interaction, we present a possible
trace of the evolution of the constraints. The relation $\goesby{~~}^{*}$ denotes the
evolution of the constraints by either applying transformations that preserve
the set of possible solutions, or by extending the interpretation \I based on
external interaction. The initial persistent and 
ephemeral constraints of each primitive $p\in P$ are denoted by $\rho_p$ and
$\epsilon_p$, respectively.
As in Example~\autoref{ex:lossymerge}, $\phi_i$, where $i\in 1..6$, are the
constraints for CDG$_1$, CDG$_2$, $\lossy_1$, $\lossy_2$, $\mathit{Merger}$,
and Client, respectively.
\[\begin{array}{rcl@{~~~~~}rcl}
\rho_{\textrm{CDG}_1} &=& SFA(a)  &
  \epsilon_{\textrm{CDG}_1} &=& \phi_1 \\
\rho_{\textrm{CDG}_2} &=& SFA(c)  &
  \epsilon_{\textrm{CDG}_2} &=& \phi_2 \\
\rho_{\textrm{Client}} &=& \phi_6 \sland SFA(e)  &
  \epsilon_{\textrm{Client}} &=& \true \\
\end{array}
~~~~~~~~~~~~~~~~~
\begin{array}{rcl@{~~~~~}rcl}
\rho_{\lossy_1} &=& \phi_{3} \sland SFA(a,b)  &
  \epsilon_{\lossy_1} &=& \true \\
\rho_{\lossy_2} &=& \phi_{4} \sland SFA(c,d)  &
  \epsilon_{\lossy_2} &=& \true \\
\rho_{\mathit{Merger}} &=& \phi_{5} \sland SFA(b,d,e)  &
  \epsilon_{\mathit{Merger}} &=& \true \\
\end{array}
\]
The initial configuration of the system is given by the set $\block{\rho_p
\land \epsilon_p}{p}^{p \in P}$. We write $\epsilon_p^n$ to denote the ephemeral
constraint of $p$ in round $n$. During the execution of the engine, both the
constraints and the interpretation changes during the solving stage, which we
make explicit by using a pair with the interpretation and the constraints. The
evolution of a possible trace for our example and its explanation follows:
\[
\begin{array}{l}
1\left\{
  \begin{array}{r@{~~}l}
    &\I,\block{\rho_p \land \epsilon_p^1}{p}^{p \in P}
    \\
    \goesby{~~}^{*}
    &\I,\block{\phi_1\land\phi_3\land\phi_5\land\phi_6\land  
        SFA(\set{a,b,d,e})}{Q},
    \block{\phi_2\land\phi_4\land SFA(\set{c,d})}{R}
    \\
    \goesby{~~}^{*}
    &\I,\block{(a \land b \land e \land \evar{more}(\wh{a}) \land
     \wh{b}=\wh{a} \land \wh{e}=\wh{b} \land \evar{result}=\wh{e}) \slor 
     \psi_{q}}{Q},
     \block{(c \land \wh{c}=\mathtt{d_2} \land \lnot{d}) \slor \psi_{r}}{R}
  \end{array}
\right.
\vspace*{2mm}\\
2\left\{
  \begin{array}{r@{~~}l}
    \,\,\goesby{~~}^{*}
    &\I,\block{(a \land b \land e \land \evar{more}(\wh{a}) \land
     \wh{b}=\wh{a} \land \wh{e}=\wh{b} \land \evar{result}=\wh{e}) \slor \psi_{q}}{Q},
     \block{\rho_r \land \epsilon_r^2}{r}^{r \in R}
  \end{array}
\right.
\vspace*{2mm}\\
3\left\{
  \begin{array}{r@{~~}l@{~}l}
   \goesby{~~}^{*}
    &\I',&\block{(a \land b \land e \land \wh{a}=\mathtt{d_4} \land 
     \evar{evenmore}(\wh{a}) \land
     \wh{b}=\wh{a} \land \wh{e}=\wh{b} \land \evar{result}=\wh{e}) \slor \psi_{q}}{Q},
     \\&&
     \block{\rho_r \land \epsilon_r^2}{r}^{r \in R}
    \\
    \goesby{~~}^{*}
    &\I',&\block{(a \land b \land e \land \wh{a}=\mathtt{d_4} \land
     \wh{b}=\mathtt{d_4} \land \wh{e}=\mathtt{d_4} \land \evar{result}=\mathtt{d_4}) \slor 
     \psi'_{q}}{Q},
     \block{\rho_r \land \epsilon_r^2}{r}^{r \in R}
  \end{array}
\right.
\vspace*{2mm}\\
4\left\{
  \begin{array}{r@{~~}l}
    \,\,\goesby{~~}^{*}
    &\I,
     \block{\rho_q \land \epsilon_q^2}{q}^{q \in Q},
     \block{\rho_r \land \epsilon_r^2}{r}^{r \in R}
  \end{array}
\right.
\end{array}
\]

We now look in more detail into each of the transitions applied above.
\begin{enumerate}

  \item The blocks of constraints are joined into two blocks based on the 
	partition $Q$ and $R$ (of $P$). The persistent and
  ephemeral constraints are replaced by their definitions. 
  The constraints inside each new block are manipulated 
  following traditional constraint solving techniques until we obtain
  a disjunction of cases. We then focus on one specific disjunct in each block.

  \item The block tagged with $R$ (in the last step of (1)) 
 has a trivial solution $\set{c\mapsto \true,
  d \mapsto \false, \wh{c} \mapsto \mathtt{d_2}}$ that does not cause any
  state change. As the boundary conditions hold ($d \neq \true$),  we
  can perform the update phase on this block. Hereafter  the individual blocks
  for each primitive $r\in R$ are restored, updating the ephemeral constraints to
  $\epsilon_r^2$. In this case there is no state change, i.e., $\epsilon_r^2 =
  \epsilon_r^1$.
  
  \item Interaction with the external world is performed to extend the
  interpretation for
  $\evar{more}$,  obtaining $\I' = \I
  \cup \set{\evar{more} \mapsto \lambda(v). (v=\mathtt{\mathtt{d_4}} \lor
  \evar{evenmore}(v))}$. External predicate $\evar{more}(\wh{a})$ is replaced
  by its new interpretation, and  the manipulation of the
  constraints continues as in (1), until we find a new conjunction for the first block
  which satisfies the trivial solution.
  
  \item In the last step the update phase is performed on the first block. Note
  that the trivial solution obeys the boundary conditions ($d \neq \true$). 
 The individual blocks for each primitive $q \in Q$ are restored, using the
  corresponding persistent constraints and the new ephemeral constraints for
  round 2. In this case the ephemeral constraint for the primitive CDG$_1$
  is updated, while the other primitives in $Q$ keep the same ephemeral
  constraints. The update of the constraints of CDG$_1$ is performed by
  querying the external primitive CDG$_1$ for its new ephemeral constraints,
  providing the value of the communication variable of CDG$_1$
  ($\mathbf{result} = \mathtt{d_4}$). We call this new constraint
  $\epsilon_{CDG_1}^2$. After the update, the interpretation ``forgets" the
  value of $\evar{more}$ and is reset to $\I$.

\end{enumerate}
We leave for future work the formalisation of the rules that describe the evolution of the constraints and the interpretation during the constraint solving process.

\subsection{Discussion}

Local firings can be discovered concurrently. Furthermore, the explicit
connection introduced by the ownership relation, from blocks of constraints
and external symbols to external primitives, paves the way for
constraint-solving techniques that interact with the external world while
searching for solutions for constraints (concurrently). We start by discussing
some of our motivation to introduce the local satisfaction relation, and we
then explore some more details of our proposed interactive engine.

\subsubsection*{Why locality?}

Some of the inspiration for developing a semantic framework that takes into account
locality aspects of a model that requires global synchronisation came from
experiments undertaken during the development of a distributed implementation of
Reo.\footnote{\url{http://reo.project.cwi.nl/cgi-bin/trac.cgi/reo/wiki/Tools\#
DistributedReoEngine}} This implementation is incorporated in the Eclipse
Coordination Tools, and its distributed entities roughly correspond to primitives
in our constraint approach. There we also make a similar distinction between
the two phases of the engine. While developing the distributed engine, we
realised the following useful property of the FIFO$_1$ channels: 
in each round it is sufficient to consider the two halves of
a FIFO independently. 
This property went against the implicit globality assumption  in 
current Reo models, and was never clearly exploited by Reo.
This locality property becomes particularly relevant in the
extreme case of a \reo connector consisting of several FIFO$_1$ channels are composed
sequentially.
In the communication between any two FIFO's from this sequence, traditional
\reo models require all the FIFO's to agree, while our distributed
implementation requires only the agreement of the two FIFO's involved in the
communication.

In more complex \reo connectors, such as the \emph{multiple
merger},\footnote{\url{http://homepages.cwi.nl/~proenca/webreo/generated/
multimerge/frameset.htm}} it is possible to see that most of the steps
involve only the flow on a small part of the connector. It is also possible to
find \emph{islands} of synchronous regions, with FIFO channels in the boundaries, where
our boundary condition holds for the possible solutions.
The approach described in this paper not only justifies the correctness of the 
locality obtained by the FIFO$_1$ channels, but it also generalises it to arbitrary 
solutions where the boundary conditions hold on the boundaries of the synchronous region.

\subsubsection*{Interactive engine}

We now explore some characteristics of the engine described in \autoref{sec:engine}, using the logic with external symbols introduced in \autoref{sec:externallogic}.
We assume that the interpretation \I is initially empty regarding external symbols. During the solve stage, \I is extended every time the external world provides new information about these external symbols. Similarly, the engine can request for the interpretation of specific symbols whenever these are required to find solution.
The communication variables play a similar role to state variables. Instead of
being directly used in the next round, their value is sent to the primitive
that owns the variable, and the engine waits for new (ephemeral) constraints
from that primitive. These constraints are then used in the next round.

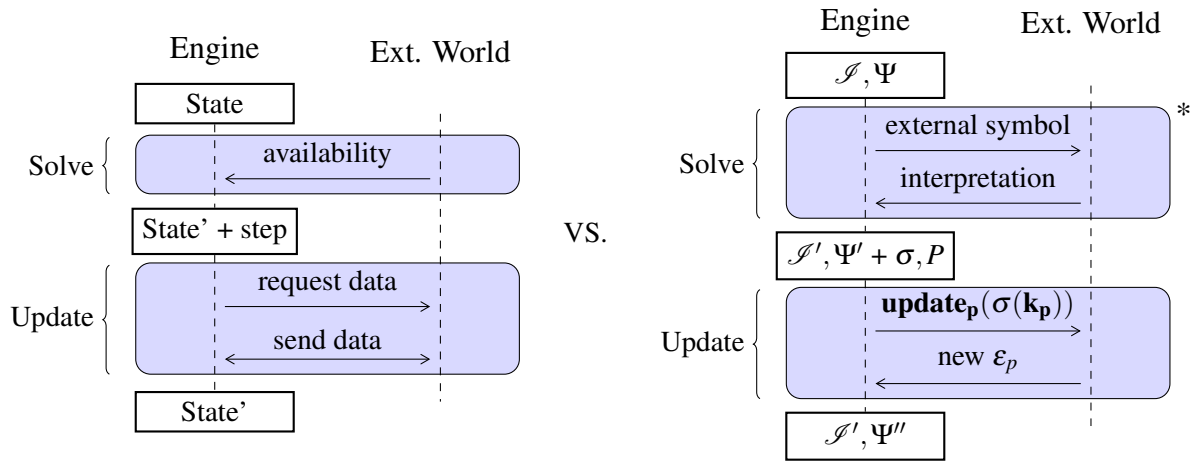
\begin{figure}[htb]
{\centering

\wrap{\begin{tikzpicture}
  [ >=angle 60,     state/.style={rectangle,minimum width=21mm,                  thick,draw=black},
                      lifeline/.style={dashed},
    title/.style={font=\large},
    node distance = 7mm,
    fststep/.style={node distance = 10mm},
    lststep/.style={node distance = 7mm},
    otherside/.style={node distance = 30mm},
    glued/.style={node distance = 0mm},
    stage/.style={blue!15,draw=black,rounded corners=2mm},
    brace/.style={midway, left=2pt}
  ]

\node[state]                 (e1)  {State}; 
\node[below of=e1,fststep]   (e2) {};
\node[below of=e2,lststep,state]
                             (e6) {State' + step}; 
\node[below of=e6,fststep]   (e7) {};
\node[below of=e7]           (e8) {};
\node[below of=e8,state]     (e9) {State'};

\node[right of=e1,otherside] (x1) {}; 
\node[below of=x1,fststep]   (x2) {}; 
\node[below of=x2,lststep]   (x6) {}; 
\node[below of=x6,fststep]   (x7) {}; 
\node[below of=x7]           (x8) {}; 
\node[below of=x8]           (x9) {}; 
\node[title,above of=e1] (eng) {Engine};
\node[title,above of=x1] (world) {Ext. World};

\draw[lifeline] (x1) -- (x9) {};
\path(e1) edge[lifeline] (e6)
     (e6) edge[lifeline] (e9);

\draw[<-,auto] (e2) -- node {availability}           (x2) {};
\draw[->,auto] (e7) -- node {request data} (x7) {};
\draw[<->,auto] (e8) -- node {send data}        (x8) {};

\begin{pgfonlayer}{background}
  \fill[stage,name=solve]
     ($(e2) + (-10.5mm,5.8mm)$) rectangle ($ (x2) + (10.5mm,-2mm) $); 
  \fill[stage,name=update]
     ($(e7) + (-10.5mm,5.8mm)$) rectangle ($ (x8) + (10.5mm,-2mm) $); 
\end{pgfonlayer} 

\draw[decorate,decoration=brace] 
   ($ (e2) + (-14mm,-2mm)$)  -- ($(e2) + (-14mm,5.5mm)$) node[brace] {Solve};
\draw[decorate,decoration=brace] 
   ($ (e8) + (-14mm,-2mm)$)  -- ($(e7) + (-14mm,5.5mm)$) node[brace] {Update};

\end{tikzpicture}}
~~~~
VS.
~~~~
\wrap{\begin{tikzpicture}
  [ >=angle 60,     state/.style={rectangle,minimum width=21mm,                  thick,draw=black},
                      lifeline/.style={dashed},
    title/.style={font=\large},
    node distance = 7mm,
    fststep/.style={node distance = 10mm},
    lststep/.style={node distance = 7mm},
    halfstep/.style={node distance = 5mm},
    otherside/.style={node distance = 30mm},
    stage/.style={blue!15,draw=black,rounded corners=2mm},
    brace/.style={midway, left=2pt}
  ]

\node[state]                 (e1)  {$\I,\Psi$}; 
\node[below of=e1,fststep]   (e4) {};
\node[below of=e4]           (e5) {};
\node[below of=e5,lststep,state]
                             (e6) {$\I',\Psi'$ + $\sigma,P$}; 
\node[below of=e6,fststep]   (e7) {};
\node[below of=e7]           (e8) {};
\node[below of=e8,state]     (e9) {$\I',\Psi''$};

\node[right of=e1,otherside] (x1) {}; 
\node[below of=x1,fststep]   (x4) {}; 
\node[below of=x4]           (x5) {}; 
\node[below of=x5,lststep]   (x6) {}; 
\node[below of=x6,fststep]   (x7) {}; 
\node[below of=x7]           (x8) {}; 
\node[below of=x8]           (x9) {}; 
\node[title,above right of=x2] (star) {~~~~~~~~~~~~~~*};
\node[title,above of=e1] (eng) {Engine};
\node[title,above of=x1] (world) {Ext. World};

\draw[lifeline] (x1) -- (x9) {};
\path(e1) edge[lifeline] (e6)
     (e6) edge[lifeline] (e9);

\draw[->,auto] (e4) -- node {external symbol}   (x4) {};
\draw[<-,auto] (e5) -- node {interpretation}    (x5) {};
\draw[->,auto] (e7) -- node {$\mathbf{update_p(\sigma(\evar{k}_p))}$} (x7) {};
\draw[<-,auto] (e8) -- node {new $\epsilon_p$}        (x8) {};

\begin{pgfonlayer}{background}
  \fill[stage,name=solve]
     ($(e4) + (-10.5mm,5.8mm)$) rectangle ($ (x5) + (10.5mm,-2mm) $); 
  \fill[stage,name=update]
     ($(e7) + (-10.5mm,5.8mm)$) rectangle ($ (x8) + (10.5mm,-2mm) $); 
\end{pgfonlayer} 

\draw[decorate,decoration=brace] 
   ($ (e5) + (-14mm,-2mm)$)  -- ($(e2) + (-14mm,5.5mm)$) node[brace] {Solve};
\draw[decorate,decoration=brace] 
   ($ (e8) + (-14mm,-2mm)$)  -- ($(e7) + (-14mm,5.5mm)$) node[brace] {Update};
\end{tikzpicture}}

}
\caption{Interaction with \reo components (left) and with our view of 
components (right).}
\label{fig:interactiondiag}
\end{figure}

The interaction between components and the engine differs in our model with
respect to other descriptions of \reo, in that the components play a more
active role in the coordination, as depicted in \autoref{fig:interactiondiag}.
The usual execution of \reo~\cite{reo:tools} is also divided in two main 
steps, but the
interaction is more restricted in previous models of \reo. In the solve stage
the components attempt to
write or take a data value. In the update stage the engine requests or sends
data values, and restarts the solve stage.
In our model we blur the distinction between connectors and components.
During the solve stage
components can provide constraints with external symbols, that will only be prompted by the
engine as required. During the update stage  the engine sends the
components the values of their communication variables, if defined, and waits
for their new constraints for the next round. Our approach therefore offers components
the ability to play a more active and dynamic role during the
coordination.

\section{Conclusion and related work}
\label{sec:conclusion}

Despite Wegner's interesting perspective on coordination as constrained interaction~\cite{wegner}, little work takes this perspective literally, representing coordination as constraints. 
Montanari and Rossi express coordination as a constraint satisfaction
problem, in a similar and general way~\cite{montanari:98}. 
They view networks as graphs, and use the tile model to
distinguish between synchronisation and sequential composition of the
coordination pieces. In our approach, we explore a more concrete coordination
model, which not only captures the semantics of the \reo coordination
language~\cite{reo:primer}, but also extends it with 
a refined notion of locality and a variety of notions of
external interaction not found in Montanari and Rossi's work.

Minsky and Ungureanu took a practical approach and introduced the
Law-Governed Interaction (LGI) mechanism~\cite{lgi}, implemented by the Moses
toolkit. The mechanism targets distributed coordination of heterogeneous 
agents,
enforcing laws that are defined using constraints in a Prolog-like language.
The main innovation is the enforcement of laws by certified controllers that
are not centralised. Their laws, as opposed to our approach, are not global,
allowing them to achieve good performance, while compromising the scope of 
the constraints. Our approach can express constraints globally, but can
solve them locally where possible. 

In the context of the \reo coordination language, Lazovik et 
al.~\cite{laz:icsoc07}  provide a
choreography framework for
web services based on \reo, where they use constraints to solve a coordination
problem. This work is an example of a concrete application of constraints to
coordination, using a centralised and non-compositional approach. We formalised
and extended their ideas in our work on Deconstructing Reo~\cite{reo:deconstructing}.
The analogy between
\reo constraints and constraint solving problems is also pursued by 
Kl\"uppelholz and Baier~\cite{reo:mc:baier06}, who describe a symbolic model
checking for \reo, and by Maraikar et al.~\cite{reo:mashups}, who present a
service composition platform based on \reo using a mashup's data-centric
approach. The latter can be seen as an scenario where constraint solving
techniques are used for executing a \reo-based connector.

One of the main novelties with respect to our previous
work~\cite{reo:deconstructing} is the introduction of a partial semantics to
the logic, and techniques for exploiting this semantics. 
Partiality favours solutions that address
only a relevant subset of variables, and can furthermore  capture solutions
in only part of a network, which cannot be considered independently in a
classical setting. Other applications of partial or 3-valued
logic exist~\cite{partiallogic,grumberg:07}, and model checking
and SAT-based algorithms exist for such logics.
We do not address verification of partially defined systems, but instead we
focus on the specification and execution of these systems. Verification of
systems specified by a partial logic would require to assume a fixed
interpretation of external symbols, but still presents an interesting
challenge, which is out of the scope of this paper.
Note that the constraint solving of our partial logic is different from the
partial constraint satisfaction problem
(PCSP)~\cite{pcsp},
which consists of
finding solutions satisfying a constraint problem $P$ that are as close as
possible to the original problem, although they may be  different in most cases.

Faltings et al.~\cite{faltings:05} explore interactive constraint satisfaction,
which bears some similarity to our approach. They present a
framework of open constraint satisfaction in a distributed environment, where
constraints can be added on-the-fly. They also consider weighted constraints
to find optimal solutions. In this paper we do not explore strategies to make
the constraint solving process more efficient, such as considering  the order in which the 
rules should be applied. The main differences between our work and theirs are that we
focus on the coordination of third parties, making a clear distinction between
computation and coordination, we use a partial logic, and we
have more modes of interaction. 

CRIME (Consistent Reasoning in a Mobile Environment) is an implementation of
the Fact Space Model~\cite{factspacemodel}, which addresses highly interactive
programs in a constantly changing environment. Applications publish their
\emph{facts} in a \emph{federated fact space}, which is a  tuple space shared by
nearby devices. Each fact is defined as a Prolog-like constraint, and the
federated fact space evolves as other applications connect or disconnect. The
resulting system is a set of reactive objects whose topology is
constantly changing. Many of the fact space model ideas are orthogonal to the
interaction constraints described in this paper, and its implementation 
could form  a possible base platform for our approach.

\subsection*{Conclusion}

The key contributions of our work are the use of 
a local logic which does not require all constraints to be satisfied, 
and the different modes of interaction. Together these  enable more 
concurrency, more flexibility, and more scalability, providing a solid
theoretical basis for constraint satisfaction-based coordination models.
Furthermore,
constraints provide a flexible framework in which it may be possible
to combine  other constraint based notions, such as service-level
agreements. As future work we plan to explore the extension of \reo-based
tools, and to implement an interactive and iterative constraint-solving
process based on the logic described in this paper.  In the process, we 
will introduce rules describing how to manipulate blocks of constraints 
that preserve simple solutions, in order to describe in more detail the 
concurrent search for local firings. Later we plan to explore strategies 
for the application of these rules, and to understand better
the efficiency of our approach.

\bibliographystyle{eptcs}
\bibliography{bibliography}

\end{document}